\documentclass{eptcs}
\providecommand{\event}{PLACES 2013}
\usepackage{breakurl}

\title{Session Types in Abelian Logic}
\author{Yoichi Hirai\thanks{When this work was presented
at the workshop, the author was a student at the University
of Tokyo and a JSPS fellow supported by Grant-in-Aid for
JSPS Fellows 23-6978.}
\institute{National Institute of Advanced Industrial Science and Technology}
\email{y-hirai@aist.go.jp}}
\def\titlerunning{Session Types in Abelian Logic}
\def\authorrunning{Y. Hirai}

\usepackage{stmaryrd,setspace,float}
\usepackage[counterclockwise]{rotating}
\usepackage{enumitem}
\usepackage{textcomp}
\usepackage{makeidx}
\usepackage{caption}
\usepackage{tipa}

\newcommand{\ignore}[1]{}

\usepackage{rotating}

\usepackage{bussproofs,stmaryrd}
\EnableBpAbbreviations

\usepackage{graphicx}
\usepackage{amssymb,amsmath,amsthm}
\usepackage{amscd,cmll}

\newtheorem{theorem}{Theorem}[section]
\newtheorem{corollary}[theorem]{Corollary}
\newtheorem{example}[theorem]{Example}
\newtheorem{proposition}[theorem]{Proposition}

\newtheorem{definition}[theorem]{Definition}
\newtheorem{lemma}[theorem]{Lemma}

\makeatletter
\def\NAT@spacechar{~}
\makeatother

\newcommand{\tuple}[1]{\langle{#1}\rangle}
\newcommand{\fix}[1]{[FIX \fbox{#1}]}

\newcommand{\ruleskip}{\vskip 4mm}

\newcommand{\hyper}{\mathcal{H}}
\newcommand{\hypert}{\mathcal{O}}
\newcommand{\hmid}{\ \ \rule[-2pt]{2pt}{11pt}\ \ }

\newcommand{\tr}{\vdash}

\newcommand{\lpair} [1]{\langle{#1}\rangle}

\newcommand{\inl}[1]{\mathsf{inl}({#1})}
\newcommand{\inr}[1]{\mathsf{inr}({#1})}

\newcommand{\mat} [5]{\mathsf{match}\,{#1}\,\mathsf{of}\, \inl{#2}. {#3}/
\inr{#4}. {#5}}

\newcommand{\tj}   [2]{ {#1} \mathord{\colon\kern -2pt}{#2} }

\newcommand{\wwedge}{\operatorname*{\bigwedge\kern -8pt \bigwedge}}

\newcommand {\G}{\Gamma}
\newcommand {\D}{\Delta}

\renewcommand{\phi}{\varphi}

\renewcommand{\vec}{\overrightarrow}

\newcommand{\limp}{\multimap}

\newcommand{\co}[1]{\bar{#1}}

\newcommand{\one}{\mathbf{1}}

\newcommand{\letin}[3]{\mathsf{let }\,{#1}\,\mathsf{ be
}\,{#2}\,\mathsf{ in }\,{#3}}
\newcommand{\ign}[2]{\mathsf{ign }\,{#1}\,\mathsf{ in }\,{#2}}

\newcommand{\eval}{\Downarrow}

\newcommand{\terminate}{\mathsf{end}}

\newcommand{\sendtype}[2]{\mathop{!}\kern -2pt{#1}\hskip 2pt{#2}}
\newcommand{\recvtype}[2]{\mathop{?}\kern -2pt{#1}\hskip 2pt{#2}}

\newcommand{\sendterm}[3]{{#1}\langle{#2}\rangle.\,{#3}}
\newcommand{\recvterm}[3]{{{#1}({#2}).\,{#3}}}

\newcommand{\leftside }[1]{\mathop\triangleright(#1)}
\newcommand{\rightside}[1]{\mathop\triangleleft(#1) }
\newcommand{\bothside}[1]{\mathop{\triangleright\triangleleft}(#1) }

\newcommand{\hypere}{\mathcal{E}}

\newcommand{\two}{\one\oplus\one}

\newcommand{\imtrans}[1]{\text{\textopencorner}{#1}\text{\textcorner}}

\date{\fix{fix}}

\begin{document}

\maketitle

\begin{abstract}
There was a PhD student who says ``%
I found a pair of wooden shoes.  I put
a coin in the left and a key in the right.
Next morning, I found those objects in the opposite shoes.''
We do not claim existence of such shoes, but propose
a similar programming abstraction in the context of typed lambda
calculi.
The result, which we call the Amida calculus, extends Abramsky's linear lambda
calculus LF and characterizes Abelian logic.
\end{abstract}

\section{Introduction}

We propose a way to unify ML-style programming
languages~\cite{milner1997definition, marlow2010haskell} and
$\pi$-calculus~\cite{milner1999communicating}.
``Well-typed expressions do not go wrong,'' said Milner~\cite{milner1978}.
However, when communication is involved, how to maintain such a
typing principle is not yet settled.
For example, Haskell, which has types similar to the ML-style types,
allows different threads to communicate using a kind of shared data
store called an MVar \texttt{mv} of
type \texttt{MVar
a}, with commands
\texttt{putMVar mv} of type \texttt{a -> IO ()} and \texttt{takeMVar mv}
of type \texttt{IO
a}.
The former command consumes an argument of type~\texttt{a} and
the consumed argument appears from the latter command.
However, if programmers make mistakes, these commands can
cause a deadlock during execution even after the program passes type
checking.

In order to prevent this kind of mistakes,
a type system can force the programmer to use both the sender and the
receiver each once.
For doing this, we use the technique of linear types.
Linear types are refinements of the ML-style intuitionistic types.
Differently from intuitionistic types,
linear types can specify a portion of program to be used
just once.
Linear types are used by Wadler~\cite{wadler2012propositions} and
Caires and Pfenning~\cite{pfenning2010} to encode session types, but our type system can
type processes that Wadler and Pfenning's system cannot.

As intuitionistic types are based on intuitionistic logic,
linear types are based on linear logic.
There are classical and intuitionistic variants of linear logics.
From the intuitionistic linear logic,
our only addition is the Amida axiom\index{axiom!Amida}\index{Amida
axiom|see{axiom}}
$(\phi\limp\psi)\otimes(\psi\limp\phi)$.
We will see that the resulting logic is identical to Abelian
logic~\cite{casari1989} up to provability of formulae.
In the Amida calculus, we can express $\pi$-calculus-like processes as macros.
From the viewpoint of typed lambda calculi, a natural way to add
the axiom
$(\phi\limp\psi)\otimes(\psi\limp\phi)$
is to add a pair of primitives $c$ and $\co c$ so that
$\cdots ct \cdots \co c u \cdots$ reduces to
$\cdots u  \cdots t \cdots$: in words,
$c$ returns $\co c$'s argument and vice versa.
In the axiom, we can substitute the singleton type $\one$ for the
general $\psi$ to
obtain an axiom standing for the send-receive communication primitive pair
$(\phi\limp \one)\otimes(\one\limp\phi)$; the left hand side
${c}$ of type ${\phi\limp\one}$ is the sending primitive and
the right hand side ${\co c}$ of type ${\one \limp\phi}$ is the receiving
primitive.
The sending primitive consumes a data of type~$\phi$ and produces a
meaningless data of unit type~$\one$.
The receiving primitive takes the meaningless data of type~$\one$ and
produces a data of type~$\phi$.

When we want to use these primitives in lambda terms,
there is one problem: what happens to $\co c(c t)$?
In this case, we do not know the output of $c$ because the output of
$c$
comes from $\bar c$'s input, which is the output of $c$.
Fortunately, we just want to know the output of $\co c$, which is the
input of $c$, that is, $t$.
In a more complicated case $\co c(\co d(c (d t)))$,
we can reason the output of $\co c$ as the input of $c$ as the output of
$d$ as the input of $\co d$ as the output of $c$ as the input of $\co c$
as the output of $\co d$ as the input of $d$, which is $t$.

Our first contribution is encoding of session types into a linear
type system.  Although the approach is similar to that of
Caires and Pfenning~\cite{pfenning2010} and
Wadler~\cite{wadler2012propositions}, the Amida calculus has an
additional axiom so that it can type some processes that Caires-Pfenning
or Wadler's type systems cannot.
In essence, the axiom allows two processes to wait for one
another and then exchange information.

Our second contribution is a side effect of our first contribution.
The type system we developed is a previously
unknown proof system for Abelian logic~\cite{casari1989}.
In this paper, we introduce
the axioms of the form $(\phi\limp\psi)\otimes(\psi\limp\phi)$
on top of IMALL, intuitionistic
multiplicative additive linear
logic.

Our third contribution is the use of conjunctive hypersequents.
Hypersequents have been around since Avron~\cite{avron91}, but
in all cases, different components in a hypersequent were interpreted
disjunctively.
In our formalization of Abelian logic,
we use conjunctive hypersequents, where different components are
interpreted conjunctively.
This is the first application of such conjunctive hypersequents.

Later in this paper, we address some issues about
consistency (Theorem~\ref{sound-to-abelian}), complicated protocols
(Section~\ref{sec:session-process}) and encoding process calculi
(Section~\ref{sec:session-process}).

\section{Definitions}

\paragraph{Types}
We assume a countably infinite set of \textit{propositional
variables}\index{propositional variable}\index{variable!propositional}, for which
we use letters $X,Y$ and so on.
We define a type~$\phi$ by BNF:
$
 \phi::=\one \mid X \mid \phi\otimes\phi\mid \phi\limp\phi\mid
 \phi\oplus\phi\mid \phi\with\phi\enspace.
$
A \textit{formula}\index{formula} is a type.
As the typing rules (Figure~\ref{fig:exchange:rules}) reveal,
$\otimes$ is the multiplicative conjunction,
$\limp$ is the multiplicative implication,
$\oplus$ is the additive disjunction and
$\with$ is the additive conjunction.

\paragraph{Terms and Free Variables}
We assume countably infinitely many variables $x, y, z, \ldots$.
Before defining terms, following Abramsky's linear lambda calculus
LF~\cite{abramsky1993computational}, we define \textit{patterns}\index{pattern}
binding sets of variables:
\begin{itemize}
 \item $\ast$ is a pattern binding $\emptyset$,
 \item $\lpair{x,\_}$ and $\lpair{\_,x}$ are patterns binding $\{x\}$,
 \item $x\otimes y$ is a pattern binding $\{x,y\}$.
\end{itemize}
All patterns are from Abramsky's LF~\cite{abramsky1993computational}.
Using patterns, we inductively define a \textit{term}\index{term}~$t$
with \textit{free variables}\index{free variable}\index{variable!free}~$S$.
We assume countably infinitely many \textit{channels}\index{channel}
with involution satisfying $\co c\neq c$ and $\co{\co c} = c$.
\begin{itemize}
 \item $\ast$ is a term with free variables~$\emptyset$,
 \item a variable $x$ is a term with free variables $\{x\}$,
 \item if $t$ is a term with free variables~$S$, $u$ is a term with
       free variables~$S'$, and moreover $S$ and $S'$ are disjoint, then $t\otimes
       u$ and
       $tu$ are terms with free variables $S\cup S'$,
 \item if $t$ and $u$ are terms with free variables $S$, then
       $\lpair{t,u}$ is a term with free variables~$S$,
 \item if $t$ is a term with free variables~$S$, then
       $\inl t$ and $\inr t$ are terms with free variables~$S$,
 \item if $t$ is a term with free variables $S\cup \{x\}$ and $x$ is not
       in $S$, then $\lambda x.t$ is a term with free variables~$S$,
 \item if $t$ is a term with free variables~$S$, $p$ is a pattern
       binding $S'$, $u$ is a term with free variables $S'\cup S''$ and equalities
       $S\cap S'' = S'\cap S'' = \emptyset$ hold, then,
       $\letin t p u$ is a term with free variables $S\cup S''$,
 \item if $t$ is a term with free variables $S$,
       $u$ is a term with free variables $S''\cup \{x\}$,
       $v$ is a term with free variables $S''\cup \{y\}$,
       $x,y\notin S''$ and $S\cap S'' = \emptyset$ hold, then
       \[
	\mat t x u y v
       \]
       is a term with free variables $S\cup S''$, and
 \item if $t$ is a term with free variables~$S$,
       then $ct$ is also a term with free variables~$S$ for any channel~$c$.
\end{itemize}
Only the last clause is original,
introducing channels, which are our communication primitives.
Note that a term with free variables $S$ is not a term with free
variables $S'$ when $S$ and $S'$ are different (even if $S$ is a subset
of $S'$).
In other words, the set of free variables $FV(t)$ is uniquely defined
for a term~$t$.
We introduce an abbreviation
\begin{align*}
 \ign \epsilon t   & \equiv t\\
 \ign {s_0,\vec s} t & \equiv \letin {s_0} \ast {(\ign {\vec s} t)}
\end{align*}
inductively for a sequence of terms~$\vec s$.
Here $\epsilon$ stands for the empty sequence.
The symbol $\mathsf{ign}$ is intended to be pronounced ``ignore.''

\paragraph{Typing Derivations}
On top of Abramsky's linear lambda calculus
LF~\cite{abramsky1993computational}\index{LF}, we add a rule to
make a closed term of type $(\phi\limp\psi)\otimes(\psi\limp\phi)$.
A \textit{context}\index{context}~$\G$ is a possibly empty sequence of
variables associated with
types where the same variable appears at most once.
A context $\tj{x}{X}, \tj{y}{Y}$ is allowed, but $\tj{x}{X}, \tj{x}{Y}$
or $\tj{x}{X},\tj{x}{X}$ is not a context.
A \textit{hypersequent}\index{hypersequent} is inductively defined as
$
 \hypert ::=\, \epsilon
 \mid\, (\G\tr\tj{t}{\phi}\hmid \hypert)
$
where $\G$ is a context.
Each $\G\tr\tj{t}{\phi}$ is called a \textit{component}\index{component}
of a hypersequent.
In this paper, we interpret the components conjunctively.
Differently from the previous
papers~\cite{avron91,Baaz01122003,avrontableau},
here, the hypersequent $\G\tr\phi\hmid \D\tr\psi$ is interpreted as the
conjunction of components:
$(\bigotimes\G\limp\phi)\otimes (\bigotimes\D\limp\psi)$ where
$\bigotimes\G$ stands for the $\otimes$-conjunction of elements of $\G$.
The conjunctive treatment is our original invention, and finding an application
of such a treatment is one of our contributions.
We name this technique the \textit{conjunctive
hypersequent}\index{hypersequent!conjunctive}.
We have to note that, for Abelian logic, there is an ordinary disjunctive
hypersequent system~\cite{metcalfe2006} that enjoys cut-elimination.
We still claim that the conjunctive
hypersequents reflect some computational intuition on
concurrently running multiple processes, all of which are supposed to
succeed (as opposed to the disjunctive interpretation where
at least one of which is supposed succeed, e.g. Hirai's calculus for G\"odel-Dummett logic~\cite{hiraiflops2012}).

The typing rules of the \textit{Amida
calculus}\index{calculus!Amida}\index{Amida calculus} are in
Figure~\ref{fig:exchange:rules}. Most rules are straightforward
modification of Abramsky's
  rules~\cite{abramsky1993computational}.
  The Sync rule is original.   Rules $\with$R and $\oplus$L are only
  applicable to singleton hypersequents.
 \begin{figure}[tb]
  \centering
  \footnotesize
  \AxiomC{}
  \LL{Ax}
  \UnaryInfC{$\tj{x}{\phi}\tr\tj{x}{\phi}$}
  \DisplayProof
  \hfill
  \AxiomC{$\hypert$}
  \AxiomC{$\hypert'$}
  \LL{Merge}
  \BinaryInfC{$\hypert\hmid\hypert'$}
  \DisplayProof
  \hfill
  \AxiomC{$\hypert\hmid\G\tr\tj{t}{\phi}\hmid\tj{x}{\phi},\D\tr\tj{u}{\psi}$}
  \LL{Cut}
  \UnaryInfC{$\hypert\hmid\G,\D\tr\tj{u[t/x]}{\psi}$}
  \DisplayProof
  \ruleskip
  \AxiomC{$\hypert\hmid\G,\tj{x}{\phi},\tj{y}{\psi},\D\tr\tj{t}{\theta}$}
  \LL{IE}
  \UnaryInfC{$\hypert\hmid\G,\tj{y}{\psi},\tj{x}{\phi},\D\tr\tj{t}{\theta}$}
  \DisplayProof
  \hfill
  \AxiomC{$\hypert\hmid \G \tr\tj t\phi \hmid \D \tr\tj u\psi\hmid \hypert'$}
  \LL{EE}
  \UnaryInfC{$\hypert\hmid \D \tr\tj u\psi \hmid \G \tr\tj t\phi \hmid \hypert'$}
  \DisplayProof
  \hfill
  %
  \AxiomC{}
  \LL{$\one$R}
  \UnaryInfC{$\tr\tj{\ast}{\one}$}
  \DisplayProof
  \ruleskip
  \AxiomC{$\hypert\hmid\G\tr\tj{t}{\phi}$}
  \LL{$\one$L}
  \UnaryInfC{$\hypert\hmid\G,\tj{z}{\one}\tr\tj{\ign z t}{\phi}$}
  \DisplayProof
  \hfill
  \AxiomC{$\hypert\hmid\G\tr\tj{t}{\phi}\hmid\D\tr\tj{u}{\psi}$}
  \LL{$\otimes$R}
  \UnaryInfC{$\hypert\hmid\G,\D\tr\tj{t\otimes u}{\phi\otimes \psi}$}
  \DisplayProof
  \hfill
  \AxiomC{$\hypert\hmid\G\tr\tj{t}{\phi}\hmid \D\tr\tj{u}{\psi}$}
  \LL{Sync}
  \UnaryInfC{$\hypert\hmid
  \G\tr\tj{ct}{\psi}\hmid \D\tr\tj{\co cu}{\phi}$}
  \noLine
  \UnaryInfC{($c$ and $\co c$ uniquely introduced here)
}
  \DisplayProof
  \ruleskip
  %
  \AxiomC{$\hypert\hmid\G,\tj{x}{\phi},\tj{y}{\psi}\tr\tj{t}{\theta}$}
  \LL{$\otimes$L}
  \UnaryInfC{$\hypert\hmid\G,\tj{z}{\phi\otimes\psi}\tr\tj{\letin{z}{x\otimes
  y}{t}}{\theta}$}
  \DisplayProof
  \hfill
  \AxiomC{$\hypert\hmid\G,\tj{x}{\phi}\tr\tj{t}{\psi}$}
  \LL{$\limp$R}
  \UnaryInfC{$\hypert\hmid\G\tr\tj{\lambda x.t}{\phi\limp \psi}$}
  \DisplayProof
  \hfill
  \AxiomC{$\hypert\hmid\G\tr\tj{t}{\phi}\hmid\tj{x}{\psi},\D\tr\tj{u}{\theta}$}
  \LL{$\limp$L}
  \UnaryInfC{$\hypert\hmid\G,\tj{f}{\phi\limp\psi},\D\tr \tj{u[(ft)/x]}{\theta}$}
  \DisplayProof
  \ruleskip
  \AxiomC{$\G\tr\tj{t}{\phi}$}
  \AxiomC{$\G\tr\tj{u}{\psi}$}
  \LL{$\with$R}
  \BinaryInfC{$\G\tr\tj{\lpair{t,u}}{\phi\with\psi}$}
  \DisplayProof
  \hfill
  \AxiomC{$\hypert\hmid\G,\tj{x}{\phi}\tr\tj{t}{\theta}$}
  \LL{$\with$L$_0$}
  \UnaryInfC{$\hypert\hmid\G,\tj{z}{\phi\with\psi}\tr\tj{\letin{z}{\lpair{x,\_}}{t}}{\theta}$}
  \DisplayProof
  \hfill
  \AxiomC{$\hypert\hmid\G,\tj{y}{\psi}\tr\tj{t}{\theta}$}
  \LL{$\with$L$_1$}
  \UnaryInfC{$\hypert\hmid\G,\tj{z}{\phi\with\psi}\tr\tj{\letin{z}{\lpair{\_,y}}{t}}{\theta}$}
  \DisplayProof
  \ruleskip
  \AxiomC{$\hypert\hmid\G\tr\tj{t}{\phi}$}
  \LL{$\oplus$R$_0$}
  \UnaryInfC{$\hypert\hmid\G\tr\tj{\inl{t}}{\phi\oplus\psi}$}
  \DisplayProof
  \hfill
  \AxiomC{$\hypert\hmid\G\tr\tj{u}{\psi}$}
  \LL{$\oplus$R$_1$}
  \UnaryInfC{$\hypert\hmid\G\tr\tj{\inr{u}}{\phi\oplus\psi}$}
  \DisplayProof
  \hfill
  \AxiomC{$\G,\tj{x}{\phi}\tr\tj{u}{\theta}
  $}
  \AxiomC{$\G,\tj{y}{\psi}\tr\tj{v}{\theta}$}
  \LL{$\oplus$L}
  \BinaryInfC{$\G,\tj{z}{\phi\oplus\psi}\tr\tj{\mat{z}{x}{u}{y}{v}}{\theta}$}
  \DisplayProof
  \ruleskip
  \caption[The typing rules of the Amida calculus]
  {The typing rules of the Amida calculus.
  $\hypert$ and $\hypert'$ stand for hypersequents.
  }
  \label{fig:exchange:rules}
 \end{figure}
 When $\tr\tj t\phi$ is derivable, the
 type~$\phi$ is \textit{inhabited}\index{inhabited}.
 \begin{example}[Derivation of the Amida axiom]
The type $(\phi\limp\psi)\otimes(\psi\limp\phi)$ is inhabited by
the following derivation.
 \begin{center}
  \AxiomC{}
  \LL{Ax}
  \UnaryInfC{$\tj{x}{\phi}\tr\tj{x}{\phi}$}
  \AxiomC{}
  \LL{Ax}
  \UnaryInfC{$\tj{y}{\psi}\tr\tj{y}{\psi}$}
  \LL{Merge}
  \BinaryInfC{$\tj{x}{\phi}\tr\tj{x}{\phi} \hmid
  \tj{y}{\psi}\tr\tj{y}{\psi}$}
  \LL{Sync}
  \UnaryInfC{$ \tj{x}{\phi}\tr\tj{cx}{\psi} \hmid \tj{y}{\psi}
  \tr\tj{\co c y}{\phi}$}
  \LL{$\limp$R}
  \UnaryInfC{$ \tr\tj{\lambda x.cx}{\phi\limp\psi} \hmid \tj{y}{\psi}
  \tr\tj{\co c y}{\phi}$}
  \LL{$\limp$R}
  \UnaryInfC{$ \tr\tj{\lambda x.cx}{\phi\limp\psi} \hmid
  \tr\tj{\lambda y.\co c y}{\psi\limp\phi}$}
  \LL{$\otimes$R}
  \UnaryInfC{$ \tr\tj{{(\lambda x.cx) \otimes (\lambda y.\co c
  y)}}{(\phi\limp\psi)\otimes(\psi\limp\phi)}$}
  \DisplayProof
 \end{center}
Another example shows how we can type the term $\co c(c x)$.
 \begin{center}
  \AxiomC{}
  \LL{Ax}
  \UnaryInfC{$\tj{x}{\phi}\tr\tj{x}{\phi}$}
  \AxiomC{}
  \LL{Ax}
  \UnaryInfC{$\tj{y}{\psi}\tr\tj{y}{\psi}$}
  \LL{Merge}
  \BinaryInfC{$\tj{x}{\phi}\tr\tj{x}{\phi} \hmid
  \tj{y}{\psi}\tr\tj{y}{\psi} $}
  \LL{Sync}
  \UnaryInfC{
  $\tj{x}{\phi}\tr\tj{cx}{\psi} \hmid
  \tj{y}{\psi}\tr\tj{\co cy}{\phi} $
  }
  \LL{Cut}
  \UnaryInfC{
  $\tj{x}{\phi}\tr\tj{\co c(cx)}{\phi}$
  }
\DisplayProof
 \end{center}
  \vskip 1mm
 \end{example}

\paragraph{Evaluation}
As a programming language, the Amida calculus is equipped with an
operational semantics that evaluates some closed hyper-terms into a sequence
of canonical forms.
The \textit{canonical forms}\index{canonical form} are the same as those of Abramsky's
LF~\cite{abramsky1993computational}:
\[
 \lpair{t,u}\qquad \ast\qquad v\otimes w\qquad \lambda
 x.t\qquad \inl{v}\qquad\inr{w}
\]
where $v$ and $w$ are canonical forms and $t$ and $u$ are terms.

An \textit{evaluation
sequence}\index{hypersequent!evaluation}\index{evaluation
hypersequent}~$\hypere$ is defined by the following grammar:
\[
 \hypere ::= \epsilon\mid (t\eval v\hmid \hypere)
\]
where $t$ is a term and $v$ is a canonical form.
Now we define evaluation as a set of evaluation sequences
(Figure~\ref{fig:eval}).
Though most rules are similar to those of Abramsky's
LF~\cite{abramsky1993computational},
we add the semantics for channels.
It is noteworthy that the results of evaluation are always canonical
forms.

 \begin{figure}
  \centering
  \AxiomC{}
  \UnaryInfC{$\ast\eval \ast$}
  \DisplayProof
  \hfill
  \AxiomC{$\hypere\hmid t\eval \ast \hmid u\eval v$}
  \UnaryInfC{$\hypere\hmid \ign t u\eval v$}
  \DisplayProof
  \hfill
  \AxiomC{$\hypere\hmid t\eval v\hmid u\eval w$}
  \UnaryInfC{$\hypere\hmid t\otimes u\eval v\otimes w$}
  \DisplayProof
  \hfill
  \AxiomC{$\hypere\hmid t\eval v\otimes w\hmid u[v/x,w/y]\eval v'$}
  \UnaryInfC{$\hypere\hmid \letin t {x\otimes y} u \eval v'$}
  \DisplayProof
  \ruleskip
  \AxiomC{$\hypere$}
  \AxiomC{$\hypere'$}
  \LL{Merge}
  \BinaryInfC{$\hypere\hmid \hypere'$}
  \DisplayProof\enspace
  (For any channel~$c$, it is not the case that $\hypere$ contains $c$
  and $\hypere'$ contains $\co c$.)
  \ruleskip
  \AxiomC{}
  \UnaryInfC{$\lambda x.t\eval \lambda x.t$}
  \DisplayProof
  \ruleskip
  \AxiomC{$\hypere\hmid t\eval\lambda x.t'\hmid u\eval v\hmid
  t'[v/x]\eval w$}
  \UnaryInfC{$\hypere \hmid tu \eval w$}
  \DisplayProof
  \hfill
  \AxiomC{$\hypere\hmid t\eval v\hmid u\eval w$}
  \UnaryInfC{$\hypere\hmid ct\eval w\hmid \co cu\eval v$}
  \DisplayProof
  ($\hypere$, $t$ and $u$ do not contain $c$ or $\co c$.)
  \ruleskip
  \AxiomC{   $\hypere\hmid t\eval t'\hmid s\eval s'\hmid
  \hypere'$}
  \UnaryInfC{$\hypere\hmid s\eval s'\hmid t\eval t'\hmid \hypere'$}
  \DisplayProof
  \hfill
  \AxiomC{}
  \UnaryInfC{$\lpair{t,u}\eval \lpair{t,u}$}
  \DisplayProof
  \hfill
  \AxiomC{$\hypere\hmid
  t\eval\lpair{t_0,t_1}
  \hmid  u[t_0/x]\eval w$}
  \UnaryInfC{$\hypere \hmid
  \letin t {\lpair{x,\_}} u\eval w$}
  \DisplayProof
  \ruleskip
  \AxiomC{$\hypere
  \hmid t\eval\lpair{t_0,t_1}
  \hmid  u[t_1/y]\eval w$}
  \UnaryInfC{$\hypere \hmid
  \letin t {\lpair{\_,y}} u\eval w$}
  \DisplayProof
  \hfill
  \AxiomC{$\hypere\hmid t\eval v$}
  \UnaryInfC{$\hypere\hmid \inl{t}\eval \inl{v}$}
  \DisplayProof
  \hfill
  \AxiomC{$\hypere\hmid u\eval w$}
  \UnaryInfC{$\hypere \hmid \inr{u}\eval \inr{w}$}
  \DisplayProof
  \ruleskip
  \AxiomC{$\hypere\hmid t\eval \inl{v}\hmid u[v/x]\eval w$}
  \UnaryInfC{$\hypere\hmid \mat t x u y {u'}\eval w$}
  \DisplayProof
  \hfill
  \AxiomC{$\hypere\hmid t\eval \inr{v}\hmid u'[v/y]\eval w$}
  \UnaryInfC{$\hypere\hmid \mat t x u y {u'}\eval w$}
  \DisplayProof
  \caption[The definition of evaluation relation of the Amida
  calculus]{The definition of evaluation relation of the Amida calculus.
  $\hypere$ is possibly the empty evaluation sequence.
  }
  \label{fig:eval}
 \end{figure}

\section{Type Safety}
\label{type-safe}

When we can evaluate a derivable hypersequent,
the result is also derivable.
Especially, this shows that, whenever a communicating term is used,
the communicating term is used according to the types
shown in the Sync rule occurrence introducing the communicating term.

 \begin{theorem}[Type Preservation of the Amida calculus]
  \label{safety}
  If terms $t_0,\ldots,t_n$ have a hypersequent
   $\tr \tj{t_0}{\phi_0} \hmid \cdots \hmid \tr \tj{t_n}{\phi_n}$ and
  an evaluation sequence $t_0\eval v_0\hmid \cdots \hmid t_n\eval v_n$
  derivable, then \\
$
   \tr\tj{v_0}{\phi_0} \hmid \cdots \hmid \tr \tj{v_n}{\phi_n}
$
  is also derivable.
 \end{theorem}
 \begin{proof}
  By induction on evaluation using the propositions below.
  We analyze the cases by the last rule.
  \begin{description}
   \item[(Merge)]
	By Proposition~\ref{split}, we can use the induction hypothesis.
   \item[($\letin t {\lpair{x,\_}} u$)]
	By Proposition~\ref{inv-with-l}, we can use the induction hypothesis.
   \item[(Other cases)]
	Similar to above.\qedhere
  \end{description}
 \end{proof}

 Two hypersequents $\hypert$ and $\hypert'$ are
 \textit{channel-disjoint}\index{channel-disjoint} if and only if
 it is not the case that $\hypert$ contains $c$ and $\hypert'$ contains
 $\co c$ for any channel~$c$.
 \begin{proposition}[Split]
  \label{split}
  If a type derivation leading to $\hypert\hmid \hypert'$ exists for two
  channel-disjoint
  hypersequents,
  both $\hypert$ and $\hypert'$ are derivable separately.
 \end{proposition}
 \begin{proof}
  By induction on the type derivation.
 \end{proof}

 \begin{proposition}[Inversion on $\with$L]
  \label{inv-with-l}
  If $\hypert\hmid \G\tr\tj{\letin t {\lpair{x,\_}} u}\theta$
  is derivable, then there is a partition of $\G$ into $\G_0$ and $\G_1$
  (up to exchange) such that $\hypert\hmid \G_0\tr \tj t {\phi\with\psi}
  \hmid \G_1,\tj{x}{\phi}\tr\tj u \theta$ is derivable.
 \end{proposition}
 \begin{proof}
   By induction on the original derivation.
 \end{proof}

  \textit{Determinacy}\index{determinacy}
  states that if $t\eval v$ and $t\eval w$ hold,
  then $v$ and $w$ are identical.
  Since our evaluation is given to possibly multiple terms at the same
  time, it is easier to prove a more general version.
  \begin{theorem}[General Determinacy of the Amida calculus]
   If
$
    t_0\eval v_0\hmid t_1\eval v_1\hmid \cdots \hmid t_n\eval v_n
$
   and
$
    t_0\eval w_0\hmid t_1\eval w_1\hmid \cdots \hmid t_n\eval w_n
$
   hold, then each $v_i$ is identical to $w_i$.
  \end{theorem}
  \begin{proof}
   By induction on the height of evaluation derivation.
   Each component in the conclusion has only one applicable rule.
   Also, the order of decomposing different components is irrelevant
   (the crucial
   condition is freshness of $c$ and $\co c$ in Figure~\ref{fig:eval}).
  \end{proof}

  \textit{Convergence}\index{convergence} would state that whenever a closed
  term~$t$ is typed $\tr\tj
  t\phi$, then an evaluation $t\eval v$ is also derivable for some
  canonical form~$v$.
  It is a desirable property so that
  Abramsky~\cite{abramsky1993computational}
  proves it for LF, but
  there are counter examples against
  convergence of the Amida calculus.
Consider a typed term $\tr\tj{c(\co c(\inl \ast))}{\one\oplus\one}$
 with no evaluation.  One explanation for the lack of evaluation is
  deadlock.  This illustrates that the current form of Amida calculus
  lacks
  deadlock-freedom.
  In order to avoid the deadlock and to evaluate this closed term,
  we can add the following eval-subst rule:
\[
  \AxiomC{$\hypere\hmid t\eval v\hmid u[v/x]\eval w$}
  \LL{eval-subst}
  \UnaryInfC{$ \hypere\hmid u[t/x] \eval w $}
  \DisplayProof,
\]
which enables an evaluation
   {$c(\co c (\inl\ast))\eval \inl\ast$}\enspace.
   Moreover,
   the eval-subst rule enables an
   evaluation $\co c[C[cv]]\eval v$, which
   reminds us of the call-with-current-continuation
   primitive~\cite{rees1986} and shift/reset primitives~\cite{danvy1990,asai2007}.
  However, adding the eval-subst rule breaks the current proof of
  Theorem~\ref{safety} (safety), but with some modifications,
  the safety property can possibly be proved.
  The main difficulty in proving the safety property
  can be seen in the form of eval-subst rule.
  When we only know the conclusion of an eval-subst occurrence,
  there are many possible assumptions involving free variables,
  all of which we must consider
  if we are to prove the type safety.

  \section{Session Types and Processes as Abbreviations}
  \label{sec:session-process}

    In order to see the usefulness of the communication primitives,
    we try implementing a process calculus and a session type system
    using the Amida calculus.

    \paragraph{Session Types as Abbreviations}

    As an abbreviation, we introduce \textit{session
    types}\index{type!session}\index{session type|see{type}}.
    Session types~\cite{interaction,honda-session} can specify a communication
    protocol over a channel.
    The following definitions and the descriptions are modification from
    Wadler's translations and descriptions of
    session types~\cite{wadler2012propositions}.  The notation here is
    different from the
    original notation by Takeuchi, Honda and Kubo~\cite{interaction}.
    \begin{align*}
     \sendtype\phi\psi&\equiv \phi\limp\psi &\text{
     output a value of $\phi$ then behave as $\psi$} \\
     \recvtype\phi\psi&\equiv \phi\otimes\psi &\text{input a value of
     $\phi$ then behave as $\psi$}\\
     \oplus\{l_i\colon \phi_i\}_{i\in I} &\equiv {\phi_0}\with
     \cdots \with {\phi_n}, \quad I = \{0,\ldots,n\} & \text{select from
     $\phi_i$ with label $l_i$}\\
     \with\{l_i\colon \phi_i\}_{i\in I} &\equiv {\phi_0}\oplus
     \cdots \oplus {\phi_n}, \quad I = \{0,\ldots,n\}& \text{offer choice of
     $\phi_i$ with label $l_i$}
     \\
     \terminate &\equiv \one &\text{terminator}
    \end{align*}
    where $I$ is a finite downward-closed set of natural numbers like
    $\{0,1,2,3\}$.
    As Wadler~\cite{wadler2012propositions} notes, the encoding looks
    opposite of what some would expect, but as
    Wadler~\cite{wadler2012propositions} explains, we are
    typing channels instead of processes.

    The grammar\hfill
$
     \phi,\psi ::= \terminate\mid X \mid \sendtype\phi\psi \mid
     \recvtype\phi\psi
     \mid \oplus\{l_i\colon\phi_i\}_{i\in I}
     \mid \with\{l_i\colon\phi_i\}_{i\in I}
$
    covers all types.
    A linear type ($\phi^\sim$ possibly with subscript) is generated by
    this grammar:
    \[
     \phi^\sim ::= \terminate\mid
     \sendtype{\psi}{\phi^\sim} \mid
     \recvtype{\psi}{\phi^\sim}
     \mid \oplus\{l_i\colon\phi^\sim_i\}_{i\in I}
     \mid \with\{l_i\colon\phi^\sim_i\}_{i\in I}
    \]

    We define duals of linear types.
    Again the definition is almost the
    same as Wadler's~\cite{wadler2012propositions} except that
    $\terminate$ is self-dual.
    \begin{align*}
     \overline{\sendtype\psi{\phi^\sim}}&= \,\recvtype\psi{\overline{\phi^\sim}}&
     \overline{\recvtype\psi{\phi^\sim}}&= \,\sendtype\psi{\overline{\phi^\sim}}\\
     \overline{\oplus\{l_i\colon \phi^\sim_i\}_{i\in I}} &=
     \with\{l_i\colon \overline{\phi^\sim_i}\}_{i\in I} &
     \overline{\with\{l_i\colon \phi^\sim_i\}_{i\in I}} &=
     \oplus\{l_i\colon \overline{\phi^\sim_i}\}_{i\in I} \\
     \overline{\terminate} &= \terminate\enspace.
    \end{align*}

    \paragraph{Processes as Abbreviations}
    We define the sending and receiving constructs of process calculi as
    abbreviations:
    \begin{align*}
     \sendterm x u t &\equiv t[(xu) /x] &\text{send $u$ through channel
     $x$ and then use $x$ in $t$} \\
     \recvterm x y t &\equiv \letin x {y\otimes x} t & \text{receive
     $y$ through channel $x$ and use $x$ and $y$ in $t$} \\
     0 &\equiv \ast &\text{do nothing}
    \end{align*}
    We have to be careful about substitution combined with process
    abbreviations.
    For example, $(\sendterm x u t)[s/x]$ is not $\sendterm s u t$
    because the latter is not defined.  Following the definition,
    $(\sendterm x u t)[s/x]$ is actually $(t[xu/x])[s/x] = t[su/x]$.
    We are going to introduce the name restriction $\nu x.t$ after
    implementing channels.

    Below, we are going to justify these abbreviations statically and
    dynamically.

    \paragraph{Process Typing Rules as Abbreviations}
    The session type abbreviation and the processes abbreviation allow
    us to use the typing rules in the next proposition.
     \begin{proposition}[Process Typing Rules: senders and receivers]
      \label{typing_process}
      These rules are admissible.
       \begin{center}
	\small
      \AxiomC{$\hypert\hmid\tj{y}{\psi}, \tj{x}{\chi}\tr\tj{t}{\phi}$}
	\LL{\rm recv}
      \UnaryInfC{$\hypert\hmid\tj{x}{\recvtype\psi\chi}\tr\tj{\recvterm x y
      t}{\phi}$}
      \DisplayProof
      \hfill
      \AxiomC{$\hypert\hmid \G,\tj{x}{\chi}\tr\tj{t}{\phi}\hmid
	\D\tr\tj{u}\psi$}
	\LL{\rm send}
      \UnaryInfC{$\hypert\hmid \G,\D,\tj{x}{\sendtype
      \psi\chi}\tr\tj{\sendterm x u t}{\phi}$}
      \DisplayProof
      \hfill
      \AxiomC{$\hypert\hmid \G\tr\tj{t}{\phi}$}
	\LL{\rm end}
      \UnaryInfC{$\hypert\hmid \G,\tj{x}{\terminate}\tr\tj{\ign x
      t}{\phi}$}
      \DisplayProof
	\hfill
	\AxiomC{}
	\UnaryInfC{$ \tr\tj 0 \one $}
	\DisplayProof\hfill\phantom{hoge}
       \end{center}
     \end{proposition}
      \begin{proof}
       Immediate.
      \end{proof}
We note that the types of variable $x$
     change in
     the rules.  This reflects the intuition of session types: the
     session type of a channel changes after some communication occurs
     through the channel.

    \begin{example}[Typed communicating terms]
     \label{ex:typed-processes}
     Using Theorem~\ref{typing_process}, we can type processes.
     Figure~\ref{fig:typed-process} contains one process, which sends a channel $y$ through $x$ and then
     waits for input in a channel~$y'$.
      \begin{figure}
       \centering
       \AxiomC{}
       \LL{Ax}
       \UnaryInfC{$\tj{z}{\two} \tr\tj z \two$}
       \LL{end}
       \UnaryInfC{$\tj{z}{\two},\tj{y}{\terminate}\tr\tj{\ign y z}\two$}
       \LL{end}
       \UnaryInfC{$\tj{z}{\two},\tj{x}{\terminate},
       \tj{y}{\terminate}\tr\tj{\ign {x,y} z}\two$}
       \LL{recv}
       \UnaryInfC{$\tj{x}{\terminate},\tj{y}{\recvtype{(\two)}\terminate}\tr
       \tj{\recvterm y z
       {\ign {x,y}  z}}{\two}$}
       \AxiomC{}
       \LL{Ax}
       \UnaryInfC{$\tj{y'}{\sendtype{(\two)}{\terminate}}\tr\tj{y'}{\sendtype{(\two)}\terminate}$}
       \LL{send}
       \BinaryInfC{$\tj{y}{\recvtype{(\two)}\terminate},\tj{x}{\sendtype{(\sendtype
       {(\two)} \terminate)}\terminate}, \tj{y'}{\sendtype {(\two)}
       \terminate}\tr\tj
       {\sendterm x y{\recvterm {y'}z{\ign {x,y} z}}}\two$}
       \DisplayProof
       \caption{A typed process.}
       \label{fig:typed-process}
      \end{figure}
     Here is another process that takes an input~$w'$ from channel~$x'$, where
     the input $w'$ itself is expected to be a channel.
     After receiving $w'$, the process puts $\inl{\ast}$ in $w'$.
      \begin{center}
       \AxiomC{}
       \LL{$\one$R}
       \UnaryInfC{$\tr\tj\ast\one$}
       \LL{end}
       \UnaryInfC{$\tj{w'}\terminate\tr\tj{\ign{w'}\ast}\one$}
       \AxiomC{}
       \LL{$\one$R}
       \UnaryInfC{$\tr\tj\ast\one$}
       \LL{$\oplus$R}
       \UnaryInfC{$\tr\tj{\inl{\ast}}{\two}$}
       \LL{send}
       \BinaryInfC{$\tj{w'}{\sendtype{(\two)}\terminate}\tr\tj{\sendterm{w'}{\inl\ast}{\ign
       {w'}\ast}}{\one}$}
       \LL{end}
       \UnaryInfC{$\tj{w'}{\sendtype{(\two)}\terminate},\tj{x'}\terminate\tr\tj{\ign{x'}{\sendterm{w'}{\inl\ast}{\ign{w'}\ast}}}{\one}$}
       \LL{recv}
       \UnaryInfC{$\tj{x'}{\recvtype{(\sendtype{(\two)}\terminate)}\terminate}\tr\tj{
       \recvterm{x'}{w'}{\ign{x'}{\sendterm{w'}{\inl\ast}{\ign{w'}\ast}}}}\one$}
       \DisplayProof
      \end{center}
    \end{example}

    \paragraph{Implementing Channels.}
    We introduced primitives $c$ and $\co c$ implementing
    the behavior specified by
    $(\one\limp\phi)\otimes(\phi\limp\one)$.
    These primitives can be seen as channels of session types
    $\recvtype\phi\terminate$ and $\sendtype\phi\terminate$.
    Indeed, $\recvtype\phi\terminate$ is $\phi\otimes\one$ (which is
    inter-derivable with $\one\limp\phi$) and $\sendtype\phi\terminate$ is
    $\phi\limp \one$.
    We can generalize this phenomenon to the more complicated session
    types\footnote{This is impossible using the ordinary linear types.}.
     \begin{proposition}[Session realizers]
      For any linear type~$\phi^\sim$\kern -2pt, the hypersequent
      $
       \tr\tj{t}{\phi^\sim}\hmid \tr\tj{u}{\overline{\phi^\sim}}
      $
      is derivable for some terms $t$ and $u$.
     \end{proposition}
      \begin{proof}
       Induction on $\phi^\sim$.
       \begin{description}
	\item[(end)] \AxiomC{} \LL{Ax }\UnaryInfC{$\tr\tj\ast\one$}
	     \AxiomC{} \LL{Ax} \UnaryInfC{$\tr\tj\ast\one$}
	     \LL{Merge}
	     \BinaryInfC{$\tr\tj\ast\one\hmid\tr\tj\ast\one$}
	     \DisplayProof is what we seek.
	\item[($\sendtype{\psi}{\phi^\sim}$)]
	     By the induction hypothesis,
	     $\tr\tj{t'}{\phi^\sim}\hmid \tr \tj{u'}{\overline{\phi^\sim}}$ is
	     derivable.  Using this, we can make the following
	     derivation:
	      \begin{center}
	      \AxiomC{}
	       \LL{Ax}
	       \UnaryInfC{$\tj{x}{\psi}\tr\tj{x}{\psi}$}
	       \AxiomC{IH}
	       \noLine
	       \UnaryInfC{$\tr\tj{t'}{\phi^\sim}\hmid \tr
	       \tj{u'}{\overline{\phi^\sim}}$}
	       \LL{Merge}
	       \BinaryInfC{$\tj{x}{\psi}\tr\tj{x}{\psi}\hmid \tr\tj{t'}{\phi^\sim}\hmid \tr
	       \tj{u'}{\overline{\phi^\sim}}$}
	       \LL{Sync}
	       \UnaryInfC{$\tj{cx}{\phi^\sim}\tr\tj{x}{\psi}\hmid
	       \tr\tj{\co ct'}{\psi}\hmid \tr
	       \tj{u'}{\overline{\phi^\sim}}$}
	       \LL{$\otimes$R}
	       \UnaryInfC{$\tj{x}{\psi}\tr\tj{cx}{\phi^\sim}\hmid
	       \tr\tj{(\co ct')\otimes u'}{\psi\otimes \overline{\phi^\sim}}$}
	       \UnaryInfC{$\tr\tj{\lambda x.cx}{\psi\limp\phi^\sim}\hmid
	       \tr\tj{(\co ct')\otimes u'}{\psi\otimes \overline{\phi^\sim}}$}
	       \DisplayProof\enspace.
	      \end{center}
	\item[($\recvtype\psi\phi$)]
	     Symmetric to the above.
	\item[($\oplus\{l_i\colon \phi_i\}$)]
	     By the induction hypothesis,
	     for each $i\in I$, we have
		   $
	      \tr\tj{t_i}{{\phi_i}}\hmid \tr\tj{u_i}{{\overline{\phi_i}}}
	     $
	     derived.  Hence derivable is
		   $
	      \tr\tj{t_i}{{\phi_i}}\hmid \tr\tj{i(u_i)}{\oplus_{j\in
	     I}
	     {\overline{\phi_j}}}
	     $
	     where $i(u_i)$ is an appropriate nesting of $\inl\cdot$,
	     $\inr\cdot$ and $u_i$.
	     Combining $|I|$ such derivations, we can derive
		   $
	     \tr\tj{\tuple{t_i}_{i\in I}}{\with_{i\in I}{\phi_i}}
	     \hmid
	     \tuple{\tr\tj{i(u_i)}{\oplus_{j\in
	     I}{\overline{\phi_j}}}}_{i\in I}
	    $
	     for a fresh natural number~$n$.
	\item[($\with\{l_i\colon \phi_i\}$)]
	     Symmetric to above.\qedhere
       \end{description}
      \end{proof}
      We call the above pair $t,u$ in the statement the \textit{session
      realizers}\index{session realizer} of $\phi^\sim$ and
      denote them by $\leftside{\phi^\sim}, \rightside{\phi^\sim}$.
      Moreover, we use $\bothside{\phi^\sim}$ to denote the pair
      $\leftside{\phi^\sim}\otimes\rightside{\phi^\sim}$.
      So far, a free variable with a linear type represented a channel
      serving the corresponding session type.
      Now, we can substitute the free variables with the session
      realizers so that the typed processes can actually communicate.
      If we have two terms that use free variables of type $\phi^\sim$ and
      $\overline{\phi^\sim}$,
      we can replace those free variables by session realizers.
       \begin{corollary}[Binding both ends of a channel]
	If
$
	 \hypert\hmid \G,\tj{x}{\phi^\sim}\tr\tj{t}{\psi}\hmid
	\D,\tj{y}{\overline{\phi^\sim}}\tr\tj{u}{\theta}
$
	is derivable, then
$
	\hypert\hmid \G\tr\tj{t[\leftside{\phi^\sim}/ x]}{\psi}
	\hmid \D\tr\tj{u[\rightside{\phi^\sim}/ y]}{\theta}
$
	is also derivable.
       \end{corollary}

       Now we can define the name restriction operator as an abbreviation:
       \[
	\nu \tj{x}{\phi^\sim}. t \equiv
	\letin{\bothside{\phi^\sim}}{x_L\otimes x_R} t
       \]
       where we assume injections $x\mapsto x_L$ and $x\mapsto x_R$
       whose images are disjoint.

       Then, in addition to Theorem~\ref{typing_process},
       more typing rules are available.
	\begin{proposition}[Process typing rule: name restriction]
	 \label{typing_connection}
	 The following typing rule is admissible.
	  \begin{center}
	   \AxiomC   {$\hypert\hmid
	   \G,\tj{x}{\phi^\sim},\tj{y}{\overline{\phi^\sim}}\tr\tj t \psi$}
	   \UnaryInfC{$\hypert\hmid
	   \G\tr\tj{\nu\tj{x}{\phi^\sim}.t[x_L/x][x_R/y]}{\psi}$}
	   \DisplayProof
	  \end{center}
	\end{proposition}

	\begin{example}[Connecting processes using session realizers]
	 Using the session realizers, we can connect the processes typed
	 in Example~\ref{ex:typed-processes}.  Indeed,
	 \begin{align*}
	  \tr &
	  \nu(\tj{x}{\recvtype{(\sendtype{(\two)}\terminate)}\terminate}).
	  \nu(\tj{y}{\sendtype{(\two)}\terminate}).
	  \\ & {\left(
	 \sendterm {x_R}{y_L}{\recvterm {y_R} z {\ign{x_R,y_R}z}}
	 \right)}
	 \otimes
	  \left(
	 \recvterm{x_L}{w'}{\ign{x_L}{\sendterm{w'}{\inl\ast}{\ign
	  {w'}\ast}}}
	  \right)
	 \colon{(\two)\otimes\one}
	 \end{align*}
	 is derivable.
	\end{example}
    Now we have to check the evaluation of the term in this example.
    For that we prepare a lemma.

    \paragraph{Process Evaluation as Abbreviation.}

    The intention of defining $\sendterm x u {t_0}$ and $\recvterm y z {t_1}$
    is mimicking communication in process calculi.
    When we substitute $x$ and $y$ with session type realizers,
    these terms can actually communicate.

    The next lemma can help us evaluate session realizers.

  \begin{lemma}
   \label{processtype}
   Let $t_0$ be a term containing a free variable $x$ and
   $t_1$ be a term containing free variables
   $y$ and $z$.
   The rule
   \begin{center}
    \AxiomC{$\hypere\hmid \leftside{\phi^\sim}\eval v'\hmid
    \rightside{\phi^\sim}\eval w'\hmid t_0[v'/x]\eval v\hmid u\eval
    u'\hmid
    t_1[u'/z][w'/y]\eval w$}
    \UnaryInfC{
    $\hypere\hmid
    \leftside{\sendtype\psi{\phi^\sim}}\eval \lambda x.cx
    \hmid
    \rightside{\sendtype\psi{\phi^\sim}}\eval u'\otimes w'
    \hmid 
    (\sendterm x u {t_0})[\lambda x.cx / x]\eval v
    \hmid
    (\recvterm y z {t_1})[u'\otimes w'/y] \eval w
    $}
    \DisplayProof
   \end{center}
   is admissible under presence of the eval-subst rule.
  \end{lemma}
  \begin{proof} By the derivation in Figure~\ref{fig:processtype}.
    \begin{figure}
     \scriptsize
     \centering
    \AxiomC{}
    \UnaryInfC{$\lambda x.cx\eval \lambda x.cx$}
    \AxiomC{}
    \UnaryInfC{$\lambda x.cx\eval \lambda x.cx$}
    \AxiomC{$\hypere\hmid \leftside{\phi^\sim}\eval v'\hmid
    \rightside{\phi^\sim}\eval w'\hmid t_0[v'/x]\eval v\hmid u\eval
    u'\hmid
    t_1[u'/z][w'/y]\eval w$}
     \doubleLine
    \UnaryInfC{$\hypere\hmid \leftside{\phi^\sim}\eval v'
    \hmid
    \rightside{\phi^\sim}\eval w'\hmid t_0[v'/x]\eval v\hmid u\eval
    u'\hmid \letin{u'\otimes w'}{z\otimes y}{t_1}\eval w$}
    \BinaryInfC{
    $\hypere\hmid \leftside{\phi^\sim}\eval v'
    \hmid
    \rightside{\phi^\sim}\eval w'\hmid t_0[v'/x]\eval v\hmid \lambda
    x.cx\eval \lambda x.cx\hmid u\eval
    u'\hmid \letin{u'\otimes w'}{z\otimes y}{t_1}\eval w$
    }
     \UnaryInfC{
     $
     \hypere\hmid \co c(\leftside{\phi^\sim})\eval u'
    \hmid
    \rightside{\phi^\sim}\eval w'\hmid t_0[v'/x]\eval v\hmid \lambda
    x.cx\eval \lambda x.cx\hmid cu\eval
    v'\hmid \letin{u'\otimes w'}{z\otimes y}{t_1}\eval w
     $
     }
     \UnaryInfC{
     $
     \hypere\hmid \co c(\leftside{\phi^\sim})\eval u'
    \hmid
    \rightside{\phi^\sim}\eval w'\hmid t_0[v'/x]\eval v\hmid (\lambda
    x.cx)u\eval v'\hmid \letin{u'\otimes w'}{z\otimes y}{t_1}\eval w
     $
     }
     \LL{eval-subst}
     \UnaryInfC{
     $
     \hypere\hmid \co c(\leftside{\phi^\sim})\eval u'
    \hmid
    \rightside{\phi^\sim}\eval w'\hmid t_0[(\lambda
     x.cx)u /x]\eval v\hmid \letin{u'\otimes w'}{z\otimes y}{t_1}\eval w
     $
     }
     \UnaryInfC{
     $
     \hypere\hmid
     \co c(\leftside{\phi^\sim})\otimes \rightside{\phi^\sim} \eval u' \otimes  w'
     \hmid t_0[(\lambda
    x.cx)u /x]\eval v\hmid \letin{u'\otimes w'}{z\otimes y}{t_1}\eval w
     $
     }
     \BinaryInfC{$
     \hypere\hmid \lambda x.cx \eval \lambda x.cx \hmid
     \co c(\leftside{\phi^\sim})\otimes \rightside{\phi^\sim} \eval u' \otimes  w'
     \hmid t_0[(\lambda
    x.cx)u /x]\eval v\hmid \letin{u'\otimes w'}{z\otimes y}{t_1}\eval w
     $}
    \DisplayProof
     \caption[Proof of Lemma~\ref{processtype}.]{Proof of Lemma~\ref{processtype}.    The conclusion is
     identical to our goal up to abbreviations. }
     \label{fig:processtype}
    \end{figure}
  \end{proof}

  \begin{example}[Evaluation of communicating processes]
   Here is an example of evaluation using the eval-subst rule.
   \begin{center}
    \AxiomC{}
    \UnaryInfC{$\leftside{\terminate}\eval \ast\hmid
    \rightside{\terminate}\eval \ast$}
    \AxiomC{}
    \UnaryInfC{$\ast\eval\ast$}
    \UnaryInfC{$\inl\ast\eval\inl\ast$}
    \AxiomC{}
    \UnaryInfC{$\ast\eval\ast$}
    \BinaryInfC{$\inl\ast\eval\inl\ast \hmid \ast\eval\ast$}
    \AxiomC{}
    \UnaryInfC{$\ast\eval\ast$}
    \UnaryInfC{$\inl\ast\eval\inl\ast$}
    \TrinaryInfC{$\leftside{\terminate}\eval \ast\hmid
    \rightside{\terminate}\eval \ast
    \hmid
    \inl\ast\eval\inl\ast \hmid \ast\eval\ast
    \hmid
    \inl\ast\eval\inl\ast
    $}
    \LL{$\diamondsuit$}
    \UnaryInfC{$\leftside{\sendtype{(\two)}\terminate}\eval \lambda x.cx\hmid
    \rightside{\sendtype {(\two)}\terminate}\eval \inl\ast\otimes\ast\hmid
    $}
    \noLine
    \UnaryInfC{\small
    $
    (\sendterm{x_L}{\inl\ast}{\ign {x_L}\ast})[\lambda x.cx/x_L]\eval
    \ast
    \hmid
    (\recvterm{x_R}z{\ign {x_R} z})[\inl\ast\otimes\ast/x_R]\eval\inl\ast
    $
    }
    \UnaryInfC{$\bothside{\sendtype{(\two)}\terminate}\eval
    \lambda x.cx\otimes (\inl\ast\otimes \ast) \hmid$}
    \noLine
    \UnaryInfC{\small $
    (\sendterm{x_L}{\inl\ast}{\ign {x_L}\ast})[\lambda x.cx/x_L]
    \otimes
    (\recvterm{x_R}z{\ign {x_R} z})[\inl\ast\otimes\ast/x_R]
    \eval \ast \otimes \inl\ast
    $}
    \UnaryInfC{$\nu(\tj{x}{\sendtype{(\one\oplus\one)}{\terminate}}).
    \left(\sendterm{x_L}{\inl\ast}{\ign{x_L}\ast}\right)\otimes
    \left(\recvterm{x_R}{z}{\ign {x_R}z}\right)\eval \ast\otimes\inl\ast$}
    \DisplayProof
   \end{center}
   The step $\diamondsuit$ uses Lemma~\ref{processtype}.
  \end{example}

  \begin{proposition}[Copycatting]
   For any linear type~$\phi^\sim$,
   we can derive
   $\tj{x}{\phi^\sim},\tj{y}{\overline{\phi^\sim}}\tr
   \tj{t}{\one}$
   for some term~$t$.
  \end{proposition}
  \begin{proof}
   By induction on $\phi^\sim$.
  \end{proof}

\subsection{Correctness with Respect to Abelian Logic}
We compare the Amida calculus and Abelian logic and
discover the fact that they are identical.

\begin{theorem}[Completeness of the Amida Calculus for Abelian
 Logic]
 \label{complete-to-Abelian}
 A formula is a theorem of Abelian logic if and only if the formula
 is inhabited in the Amida calculus.
\end{theorem}
\begin{theorem}[Soundness of the Amida Calculus for Abelian
 Logic]
 \label{sound-to-abelian}
 An inhabited type in the Amida calculus is a theorem of Abelian logic.
\end{theorem}
\begin{proof}
Proofs appear in the author's thesis~\cite{hirai-thesis}.
\end{proof}
Now we can use
some previous literature (Meyer and Slaney~\cite{meyer-slaney-1989}
and Casari~\cite{casari1989}) to find
some facts.
\begin{corollary}[Division by two]
 If $\phi\otimes\phi$ is inhabited, so is $\phi$.
\end{corollary}
\begin{corollary}
 The law of excluded middle $\phi\otimes(\phi\limp \one)$ 
and prelinearity $(\phi\limp\psi)\oplus(\psi\limp\phi)$
 are inhabited
 in the Amida calculus.
\end{corollary}

\section{Related Work}

Metcalfe, Olivetti and Gabbay~\cite{metcalfe2006} gave a hypersequent calculus for Abelian logic and
proved cut-elimination theorem for the hypersequent calculus.
His formulation is different from ours because Metcalfe's system does
not use conjunctive hypersequents.
Shirahata~\cite{shirahata} studied the multiplicative fragment of
Abelian logic,
which he called CMLL (compact multiplicative linear logic).
He gave a categorical semantics for the proofs of a sequent calculus
presentation of CMLL and then proved
that the cut-elimination procedure of
the sequent calculus preserves the semantics\footnote{
Ciabattoni, Stra{\ss}burger and
Terui~\cite{expanding} already pointed out
the fact that Shirahata~\cite{shirahata} and
Metcalfe, Olivetti and Gabbay~\cite{metcalfe2002} studied the same
logic.}.

Kobayashi, Pierce and Turner~\cite{kobayashi-pierce-turner} developed a
type system for the $\pi$-calculus processes.
Similarly to the type system presented here,
their type system can specify
types of communication contents through a name and
how many times a name can be used.
In some sense,
that type system is more flexible than the one shown in this paper;
their type system allows
multiple uses of a channel, replicated processes and 
weakening~\cite[Lemma~3.2]{kobayashi-pierce-turner}.
In other respects,
the type system in \cite{kobayashi-pierce-turner} is less expressible.
That type system does not have lambda abstractions.
Also, in contrast to our type system,
it is impossible to substitute a free variable with a process in that
type system.

Caires and Pfenning~\cite{pfenning2010} provide a type system for a
fragment of $\pi$-calculus.
Their type system is, on some processes,
more restrictive than the Amida calculus.
For example, this escrowing process $P$ below is not typable in their type
system:
$
 P = \sendterm x y{\recvterm x a
 {\sendterm y a 0}}\enspace.
$
The process first emits a channel~$y$ through channel~$x$ and then
takes an input from $x$ and outputs it to $y$.
Following the informal description of types by Caires and
Pfenning~\cite{pfenning2010},
the process~$P$ should be typable as
$
 \vdash P::x:(A\limp \one)\otimes A\enspace.
$
However, such typing is not possible because $(A\limp\one)\otimes A$ is
not a theorem of dual intuitionistic linear logic (DILL), which
their type system is based on.
In our type system, the following sequent is derivable
\begin{align*}
\tj{x}{{(\recvtype{A}\terminate)}{\sendtype A \terminate}}
\tr
\tj{
\nu(\tj{y}{\recvtype A \terminate}).
{\sendterm x{y_L}{\recvterm x a {{{\sendterm
{y_R} a {\ign{x,y_L,y_R}0}}}}}}
}{\one}
\end{align*}
The resulting sequent indicates that the process is typable with one
open channel~$x$ that first emits
a channel that one can receive~$A$ from, and second sends a value of
$A$.
This concludes an example of a term which our type system can type but
the type system in Caires and Pfenning~\cite{pfenning2010} cannot.
However, we cannot judge their type system to be too restrictive because
we have not yet obtained both
type safety and deadlock-freedom of Amida calculus at the same time.

On the other hand,
the most complicated example in Caires and Pfenning~\cite{pfenning2010},
which involves a drink server, directs us towards a useful extension of
the Amida calculus.
 \begin{example}[Drink server from Caires and Pfenning~\cite{pfenning2010} in the Amida
  cal.]
  \begin{align*}
   ServerProto &= (N\limp I \limp (N\otimes \one))\with (N\limp( I
  \otimes \one)) \\
   &= (\sendtype N {\sendtype { I} {\recvtype N \terminate}}) \with
   (\sendtype N {\recvtype { I} \one})
  \end{align*}
  $N$ stands for the type of strings and $I$ stands for the type of
  integers, but following Caires and Pfenning~\cite{pfenning2010}, we identify both $N$ and
  $I$ with $\one$.  Below, $SP$ abbreviates $ServerProto$.
  Here is the process of the server, which serves one client and
  terminates.\\
  \[
   Serv = \lpair{\recvterm s {pn} {\recvterm s {cn} {\sendterm s {rc}
  {\ign {pn, cn, s} 0}}}
  ,\quad
  \recvterm s {pn} {\sendterm s {pr} {\ign{s,pn} 0}}
  }
  \]
  We can derive a sequent $\tj{s}{\overline{SP}}\tr\tj{Serv}\one$.
  Here is one client:
  \begin{center}
   \AxiomC{}
   \UnaryInfC{$\tr\tj{0}\one$}
   \UnaryInfC{$\tj{s}{\terminate}\tr\tj{\ign s 0}\one$}
   \UnaryInfC{$\tj{s}{\terminate},
   \tj{pr}{I}\tr\tj{\ign{pr,s} 0}\one$}
   \UnaryInfC{
   $ \tj{s}{\recvtype {I} {\terminate}} \tr\tj{
   \recvterm s {pr} {\ign{pr,s} 0}
   }{\one}$
   }
   \AxiomC{}
   \UnaryInfC{$ \tr\tj{tea}{N} $}
   \BinaryInfC{
   $ \tj{s}{\sendtype N {\recvtype {I} \terminate}} \tr\tj{
   \sendterm s {tea}
   {\recvterm s {pr} {\ign{pr,s} 0}}
   }{\one}$
   }
   \UnaryInfC{
   $ \tj{s}{ServerProto} \tr\tj{
   \letin s {\lpair{\_, s}} {
   \sendterm s {tea}
   {\recvterm s {pr} {\ign{pr,s} 0}}}
   }{\one}$
   }
   \DisplayProof
  \end{center}
  In words, the client first chooses the server's second protocol, which
  is price quoting, and asks the price of the tea, receives the price
  and terminates.
  We can combine the server with this client.
  However, since the Amida calculus lacks the exponential modality, Amida calculus
  cannot type any term with $!ServerProto$, which the type system of
  Caires and Pfenning can~\cite{pfenning2010}.
  In order to do that, we might want to tolerate inconsistency and
 add $\mu$ and $\nu$ operators
  from the modal $\mu$-calculus, like Baelde~\cite{mumall} did, and express
  $!ServerProto$ as $\nu X. (SeverProto\otimes X)$.
  \end{example}

Wadler~\cite{wadler2012propositions} gave a type system for a process
calculus based on classical linear logic.
Although the setting is classical, the idea is more or less the same as
Caires and Pfenning~\cite{pfenning2010}.
Wadler's type system cannot type the escrowing process above.

Giunti and Vasconcelos~\cite{giunti2010} give a type system for
$\pi$-calculus with the type preservation theorem. Their type system is
extremely similar to our type system although the motivations are different;
their motivation is process calculi
 while our motivation is computational interpretation of a logic.
 It will be worthwhile to compare their system with our type system closely.

 \section{Some Future Work and Conclusion}

\paragraph{Implementation.}
Since Abelian logic is incompatible with contraction or weakening,
straightforward implementation the Amida calculus on top of Haskell or OCaml
would not be a
good way to exploit the safety of the Amida calculus.
One promising framework on which to implement the Amida calculus is linear ML%
\footnote{There are no
publications but an implementation is available at
\texttt{https://github.com/pikatchu/LinearML}\enspace.},
whose type system is based on linear logic.
Another way is using the type level programming technique of Haskell.
Imai, Yuen and Agusa~\cite{DBLP:journals/corr/abs-1110-4163} implemented
session types on top
of Haskell using the fact that Haskell types can
contain arbitrary trees of symbols; thus we should be able to use the
same technique to encode the types of Amida calculus in the Haskell types.

\paragraph{Adding Modalities.}
A tempting extension is to add modalities representing agents
and then study the relationship with the
multiparty session types~\cite{sync-multi-session, async-multi-session}.

\paragraph{Cut-Elimination.}
It is easy to see that the prelinearity
$(\phi\limp\psi)\oplus(\psi\limp\phi)$ does not have a cut-free proof.
However,
since there is a cut-free deduction system for Abelian
logic~\cite{metcalfe2006},
we consider it natural to expect the same property for a suitable
extension of the Amida calculus.

  \paragraph{Continuations.}
  The eval-subst rule enables an
  evaluation $\co c[C[cv]]\eval v$, which
  reminds us of the call-with-current-continuation
  primitive~\cite{rees1986} and shift/reset primitives~\cite{danvy1990,asai2007}.
  The appearance of these classical type system primitives is not
  surprising because Abelian logic validates $((p\limp q)\limp q)\limp p$,
  which is a stronger form of the double negation elimination.
  Possibly we could use the technique of Asai and Kameyama~\cite{asai2007} to analyze
  the Amida calculus with eval-subst rule.

  \paragraph{Logic Programming.}
  There are at least two ways to interpret logics computationally.
  One is proof reduction, which is represented by $\lambda$-calculi.
  The other is proof searching.
  We have investigated the Amida calculus, which embodies the proof reduction
  approach to the Amida axiom.
  Then what implication does
  The Amida axiom have in
  the proof searching approach?
  Let us cite an example from Kobayashi and Yonezawa~\cite[A.2]{kobayashi-yonezawa}:
  \begin{quote}
   Consumption of a message $m$ by a process $m\limp B$ is represented by
   the following deduction:
   \[
    (m\otimes(m\limp B)\otimes C)\limp (B\otimes C)
   \]
   where $C$ can be considered as other processes and messages, or an environment.
  \end{quote}
  Using the Amida axiom,
  the inverse
  \[
   (B\otimes C)\limp (m\otimes (m\limp B))\otimes C
  \]
  is derivable.
  This suggests that the Amida axiom states that some
  computation is reversible in the context of proof searching.
  We suspect that this can be
  useful within the realm of reversible computation~\cite{revcon}.

 \paragraph{Conclusion.}
We found a new axiomatization of Abelian logic: the Amida axiom
$(\phi\limp\psi)\otimes(\psi\limp\phi)$ on top of IMALL$^-$.
The axiomatization has an application for encoding process calculi and
session type systems.
The encoding, which we name the Amida calculus,
shows extra flexibility given by the new axiom.
In the current form, the flexibility comes with the cost of convergence.
Though there is a possible way to obtain convergence by adding
a new evaluation rule, then, it is still under investigation
whether type safety is preserved.

\paragraph{Acknowledgements.}
The author thanks Tadeusz Litak for encouragements and
information on relevant research.
The author also thanks Takeuti Izumi, who asked about changing a
disjunction~$\oplus$ into a conjunction~$\otimes$ after the author
talked about
$(\phi\limp\psi)\oplus (\psi\limp\phi)$, a variant of
which is used to model asynchronous communication
in Hirai~\cite{hiraiflops2012}.
Anonymous refrees' careful comments
and the workshop participants' direct questions improved
the presentation of this paper.

\bibliographystyle{eptcs}
\bibliography{hirai-places}

\appendix
 \section{Categorical Considerations}
 One might want to ask whether we can model the logic with
 a symmetric monoidal closed
 category~\cite{blute2004category}
   with identified isomorphisms
 $\sigma_{ABCD}\colon (A\limp B)\otimes (C\limp D)\rightarrow (A\limp D) \otimes
  (C\limp B)$, with naturality conditions.
  Before considering equality among morphisms,
  we know there is a non-trivial example.
   \begin{example}[Integer Model of {\cite[p.~107]{residuated}}]
    \label{smcc}
    The preorder formed by objects as integers and morphisms as the usual
    order among integers~$\le$
    forms a symmetric monoidal closed category with swaps
    when we interpret $\otimes$ as addition and
    $m \limp n$ as $n-m$.
   \end{example}
   On the other hand,
   if we take another formulation requiring natural isomorphisms
   $\mathcal C(C,A)\times\mathcal C(D,B) \cong \mathcal C(D,A)\times
   \mathcal C{(C,B)}$,
   only singletons can be preorder models because $\tuple{id_A,id_B}$ is
   mapped to $\tuple{f,g}$ where $f\colon A\rightarrow B$ and $g\colon
   B\rightarrow A$ for any two objects $A$ and $B$.

 A straightforward reading of evaluation rules gives somewhat complicated
 equality conditions for morphisms.
 The condition says the following diagram commutes:
 \[
    \begin{CD}
     (A\limp B)\otimes (C\limp D) @>{d_{ABCD}}>> (A\otimes
      C)\limp(B\otimes D)\\
     @VV{\sigma_{ABCD}}V @VV{id\limp s_{B,D}}V\\
     (A\limp D)\otimes (C\limp B) @>{d_{ADCB}}>> (A\otimes C)\limp (D\otimes B)
    \end{CD}
 \]
  where $d_{ABCD}$ is induced by adjunction between $\otimes$ and $\limp$
  from a morphism
  $((A\limp B)\otimes (C\limp D))\otimes (A\otimes C)\rightarrow
  (B\otimes D)$, which is provided by symmetric monoidal closed
  properties.

  Moreover, since $\phi^\ast \equiv \phi\limp\one$
  has derivable sequents $\one \tr \phi^\ast \otimes \phi$
  and $\phi\otimes\phi^\ast\tr \one$,
  we can expect the semantics of $\phi^\ast$ to be the dual object of
  that of $\phi$.
  Indeed, checking one of the coherence condition of compact closedness
  is evaluating the below typed term
  \begin{center}
   \AxiomC{$\tj t \phi$}
   \AxiomC{}
   \UnaryInfC{$\tr\tj{\ast}\one$}
   \BinaryInfC{$\tr \tj t \phi \hmid \tr \tj \ast \one$}
   \AxiomC{}
   \UnaryInfC{$\tj x \phi \tr \tj x \phi$}
   \AxiomC{}
   \UnaryInfC{$\tr \tj \ast \one$}
   \BinaryInfC{$\tj x \phi \tr \tj{cx}\one\hmid \tr \tj {\co c \ast}
   \phi$}
   \UnaryInfC{$\tj x \phi \tr \tj{cx} \one \hmid \tj{z}\one\tr \tj{\ign z
   {\co c \ast}}\phi$}
   \BinaryInfC{$\tr\tj t \phi\hmid \tj x \phi \tr \tj{cx}\one \hmid \tr
   \tj{\ign \ast {\co c \ast}}\phi$}
   \UnaryInfC{$\tr\tj{ct}\one \hmid \ign \ast {\co c \ast}\phi$}
   \UnaryInfC{$\tr\tj{ct}\one \hmid \tj y \one \tr \tj {\ign y {\ign \ast
   {\co c \ast}}} \phi$}
   \UnaryInfC{$\tr \tj{\ign {ct}{\ign \ast {\co c \ast}}}\phi$}
   \DisplayProof\enspace .
  \end{center}
  At least, if $t \eval v$ is derivable, $\ign {ct} {\ign \ast {\co c
 \ast}}\eval v$ is also derivable.
  \begin{center}
   \AxiomC{$t\eval v$}
   \AxiomC{}
   \UnaryInfC{$\ast\eval \ast$}
   \BinaryInfC{$t\eval v\hmid \ast\eval \ast$}
   \UnaryInfC{$ct \eval \ast\hmid \co c\ast \eval v$}
   \AxiomC{}
   \UnaryInfC{$\ast\eval \ast$}
   \BinaryInfC{$ct \eval \ast \hmid \ign \ast {\co c \ast} \eval v$}
   \UnaryInfC{$\ign {ct} {\ign \ast {\co c \ast}}\eval v$}
   \DisplayProof\enspace.
  \end{center}
  Showing the other direction is more
  involved because of eval-subst rule, but the author expects the case
  analysis on possible substitutions will succeed.

\section{Proof Nets}
\label{sec:proofnets}

Toward better understanding the Amida calculus, a technique called proof nets
seems promising.
Generally, proof nets are straightforward for the multiplicative
fragments but complicated when additive and exponential connectives are
involved.
Since the Amida axiom $(\phi\limp\psi)\otimes(\psi\limp\phi)$
does not contain additives ($\with,\oplus$) or exponentials ($!,?$),
we can focus on the multiplicative connectives ($\limp$ and $\otimes$).
The fragment is called \textit{IMLL} (intuitionistic multiplicative fragment of
linear logic).  We also use the unit $\one$ for technical reasons.
We first describe the IMLL proof nets and their properties.
Then we add a new kind of edges called the Amida edges, which
characterizes Abelian logic.
The Amida links are named after the Amida lottery (also known as the
Ghost Leg) for the syntactic similarity.

\subsection{IMLL Essential Nets}

The proof nets for intuitionistic linear logics are called \textit{essential
nets}\index{essential net}\index{net!essential}.
This subsection reviews some known results about the essential
nets\index{essential net} for
intuitionistic multiplicative linear logic\index{logic!intuitionistic
multiplicative linear}\index{intuitionistic multiplicative linear logic|see{logic}}
(IMLL\index{IMLL|see{logic,
intuitionistic multiplicative linear}}). The exposition here is strongly
influenced by Murawski and Ong~\cite{murawski2003}.

We can translate
a polarity $p\in\{+,-\}$ and an IMLL formula~$\phi$ into
a \textit{polarized MLL formula}\index{formula!polarized
MLL}\index{polarized MLL formula|see{formula}} $\imtrans{\phi}^p$
following Lamarche~\cite{lamarche2008} and Murawski and Ong~\cite{murawski2003}.
We omit the definition of polarized MLL formulae because the whole grammar is
exposed in the translation below:
\begin{align*}
 \imtrans{\one}^+ = \one^+ \qquad & \imtrans{\one}^- = \bot^- \\
 \imtrans{X}^+ = X^+      \qquad & \imtrans{X}^- = X^- \\
 \imtrans{\phi\limp\psi}^+ = \imtrans{\phi}^-\parr^+\imtrans{\psi}^+
 \qquad & \imtrans{\phi\limp\psi}^- = \imtrans{\phi}^+ \otimes^-
 \imtrans{\psi}^- \\
 \imtrans{\phi\otimes\psi}^+ = \imtrans{\phi}^+\otimes^+\imtrans{\psi}^+
 \qquad & \imtrans{\phi\otimes\psi}^- =
 \imtrans{\phi}^- \parr^- \imtrans{\psi}^-\enspace.
\end{align*}

For example, the Amida axiom can be translated into a polarized MLL formula
\begin{align*}
   &\imtrans{(X\limp Y)\otimes (Y\limp X)}^+ \\
 =\, &\imtrans{X\limp Y}^+ \otimes^+ \imtrans{Y\limp X}^+ \\
 =\, &\left(\imtrans{X}^- \parr^+ \imtrans{Y}^+ \right) \otimes^+
    \left(\imtrans{Y}^- \parr^+ \imtrans{X}^+ \right) \\
 =\, &\left({X}^- \parr^+ {Y}^+ \right) \otimes^+
    \left({Y}^- \parr^+ {X}^+ \right)\enspace.
\end{align*}
The symbol $\parr$ is pronounced ``parr.''

Any polarized MLL formula can be translated further into
a finite rooted tree containing these branches and polarized atomic formulae
($X^-, X^+, \one^+,\bot^-$) at the leaves.
 \begin{center}
  \includegraphics[scale=0.3]{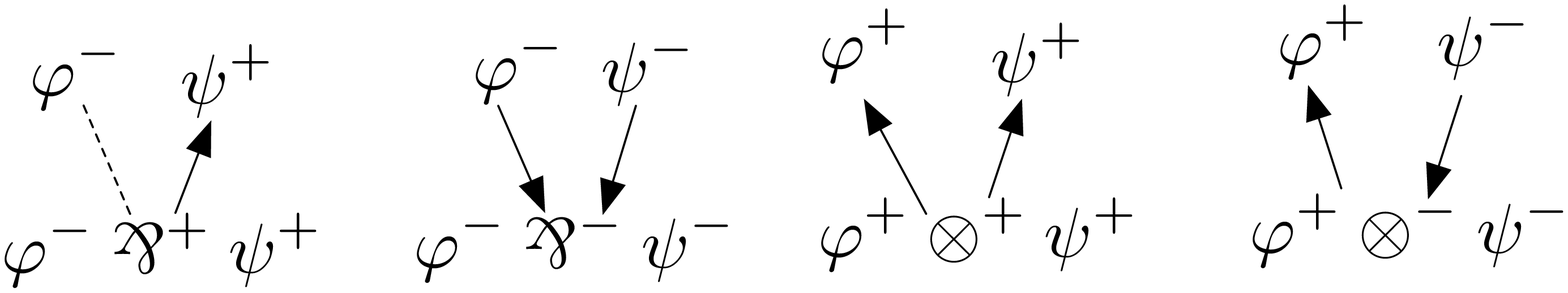}
 \end{center}
For brevity, we sometimes write only the top connectives of labeling
formulae.
In that case, these branching nodes above are denoted like this.
 \begin{center} 
  \includegraphics[scale=0.3]{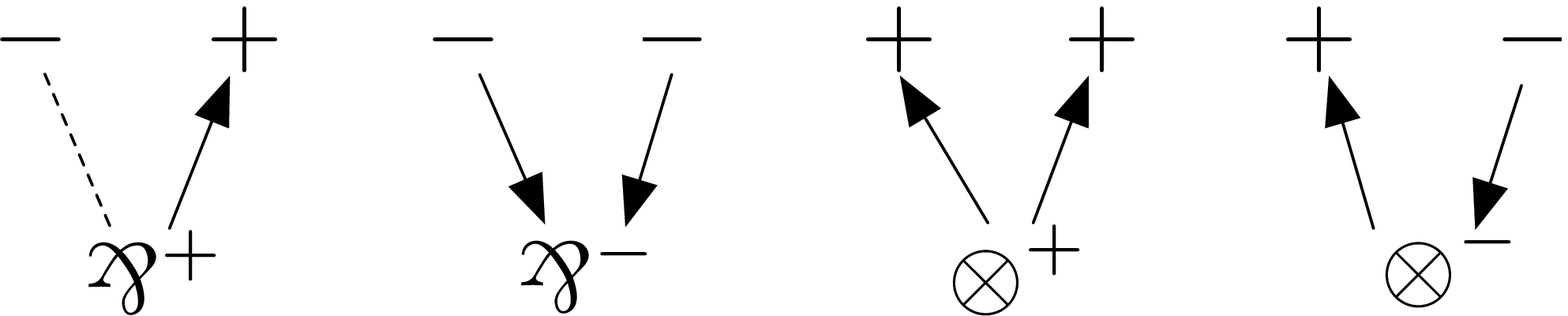}
 \end{center}
We call arrows with upward (resp. downward) signs
\textit{up-edges}\index{up-edge}
(resp. \textit{down-edges}\index{down-edge}).
The \textit{dashed child}\index{dashed child}
of a $\parr^+$ node~$p$ is the node which the dashed
line from $p$ reaches.
The branching nodes labeled by $\parr^+, \parr^-, \otimes^+$ and
$\otimes^-$ are called \textit{operator nodes}.
A \textit{path} follows
solid edges according to the direction of the edges.  Dashed edges are
not directed and they are not contained in paths.

When we add axiom edges and $\bot$-branches (shown below)
to the other operator nodes (shown above)
we obtain an \textit{essential net}\index{essential net} of $\phi$.
Due to the arbitrariness of choosing axiom edges and $\bot$-branches,
there are possibly multiple essential nets for a formula.
Murawski and Ong~\cite{murawski2003} restricts the class of formulae to linearly balanced
formulae so that the essential net is uniquely determined.
 \begin{center}
  \includegraphics[scale=0.3]{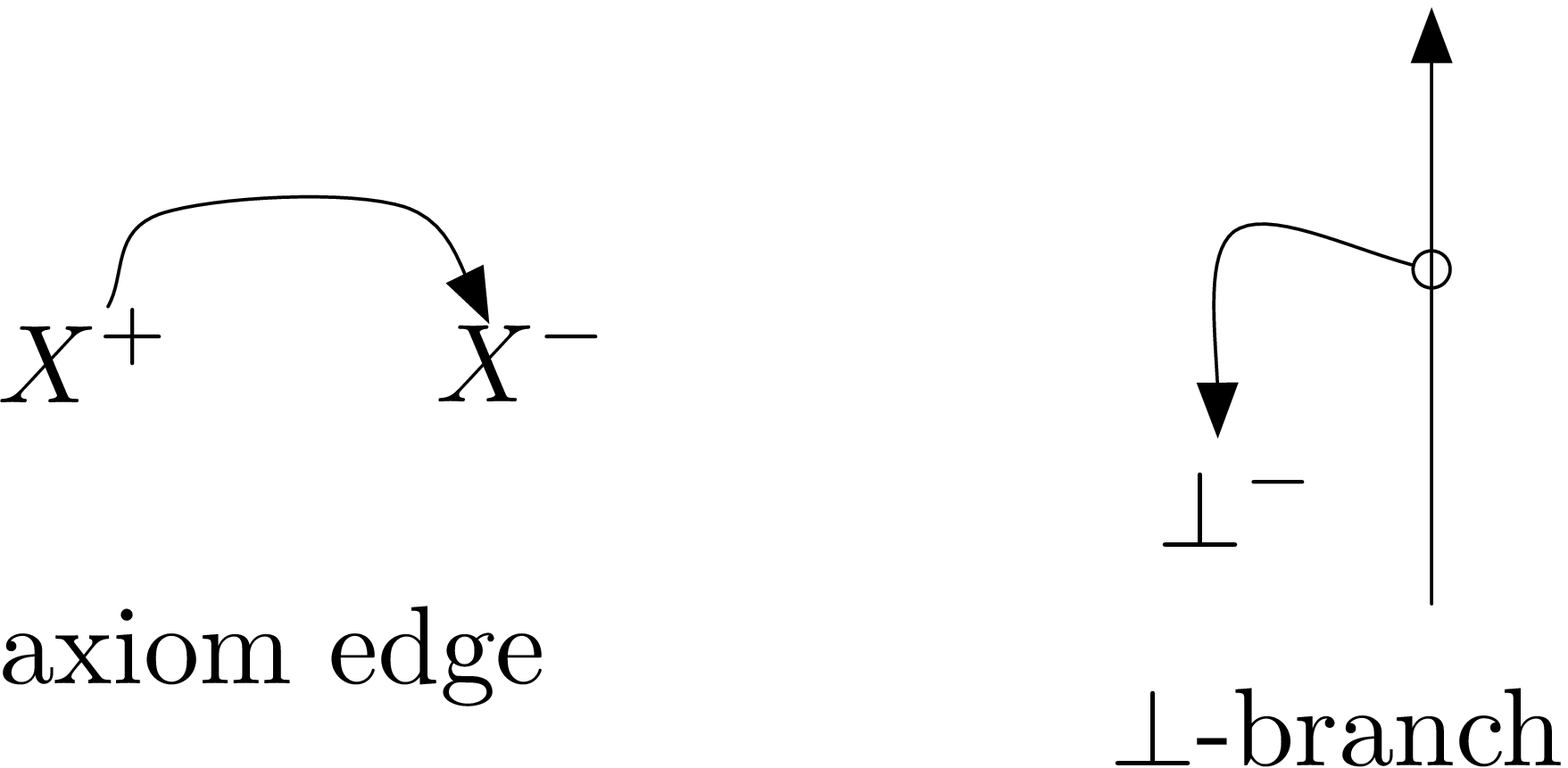}
 \end{center}

 \begin{example}[An essential net of the Amida axiom] \label{essential-amida}
  Here is one of the essential nets of
  the Amida axiom $(X^-\parr^+Y^+)\otimes^+(Y^-\parr^+ X^+)$.
   \begin{center}
    \includegraphics[scale=0.3]{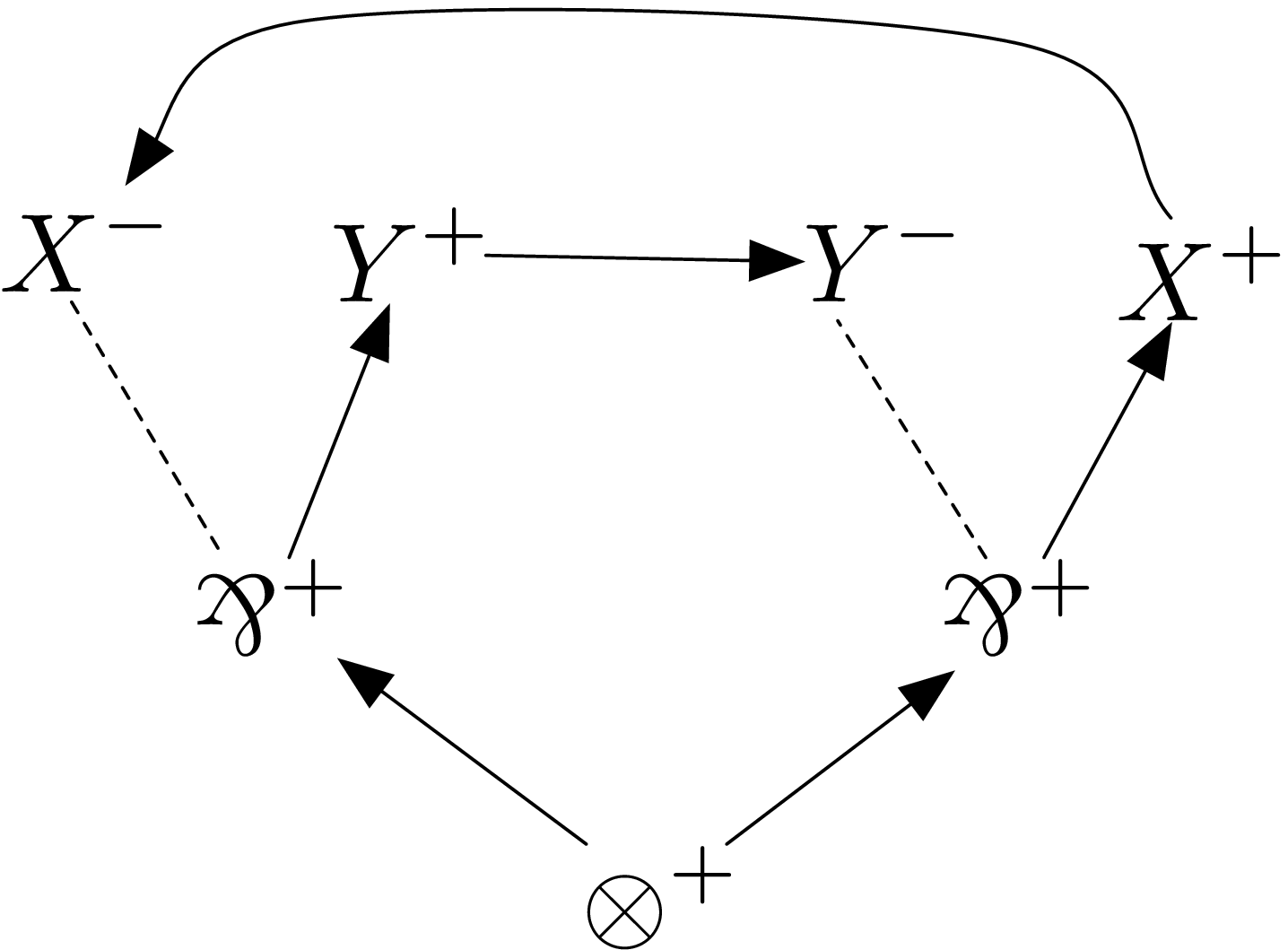}
   \end{center}
 \end{example}
 However, the essential net in Example~\ref{essential-amida} is rejected by
 the following correctness criterion.
  \begin{definition}[Correct essential nets]
   A \textit{correct essential net}\index{correct essential
   net}\index{essential net!correct} is an essential net satisfying all
   these conditions:
\begin{enumerate}
 \item Any node labeled with $X^+$ (resp. $Y^-$) is connected to a
       unique node labeled with $X^-$ (resp. $Y^+$).
       Any leaf labeled with $\bot^-$ is connected to a $\bot$-branch.
       $\one^+$ is not connected to anything above itself;
 \item the directed graph formed by up-edge, down-edge, axiom edges and
       $\bot$-branches is acyclic;
 \item \label{conditionL}
       for every $\parr^+$-node~$p$, every path
       from the root that reaches
       $p$'s dashed child also passes through $p$.
\end{enumerate}
  \end{definition}
The essential net in Example~\ref{essential-amida} is not correct for
condition~\ref{conditionL}.  Actually, the Amida axiom does not have
any correct essential net.  IMLL sequent calculus has the subformula
property so that we can confirm that the Amida axiom is not provable in IMLL.

 \begin{theorem}[Essential nets by~\cite{lamarche2008,murawski2003}]
  \label{essential-ok}
  An IMLL formula~$\phi$ is provable in IMLL if and only if there exists a correct essential net
  of $\phi$.
 \end{theorem}
 \begin{proof}
  The left to right is relatively easy.  For the other way around,
  Lamarche~\cite{lamarche2008} uses a common technique of decomposing an
  essential net from the bottom.
  Murawski and Ong \cite{murawski2003} chose to reduce the problem to sequents of special forms
  called regular.
 \end{proof}

 Actually, Lamarche~\cite{lamarche2008} also considers the cut rule (as well as
 additive operators and exponentials) in essential nets, thus
 we can include the following general axioms (as macros) and cuts (as
 primitives) and still use Theorem~\ref{essential-ok}:
 \begin{center}
  \includegraphics[scale=0.3]{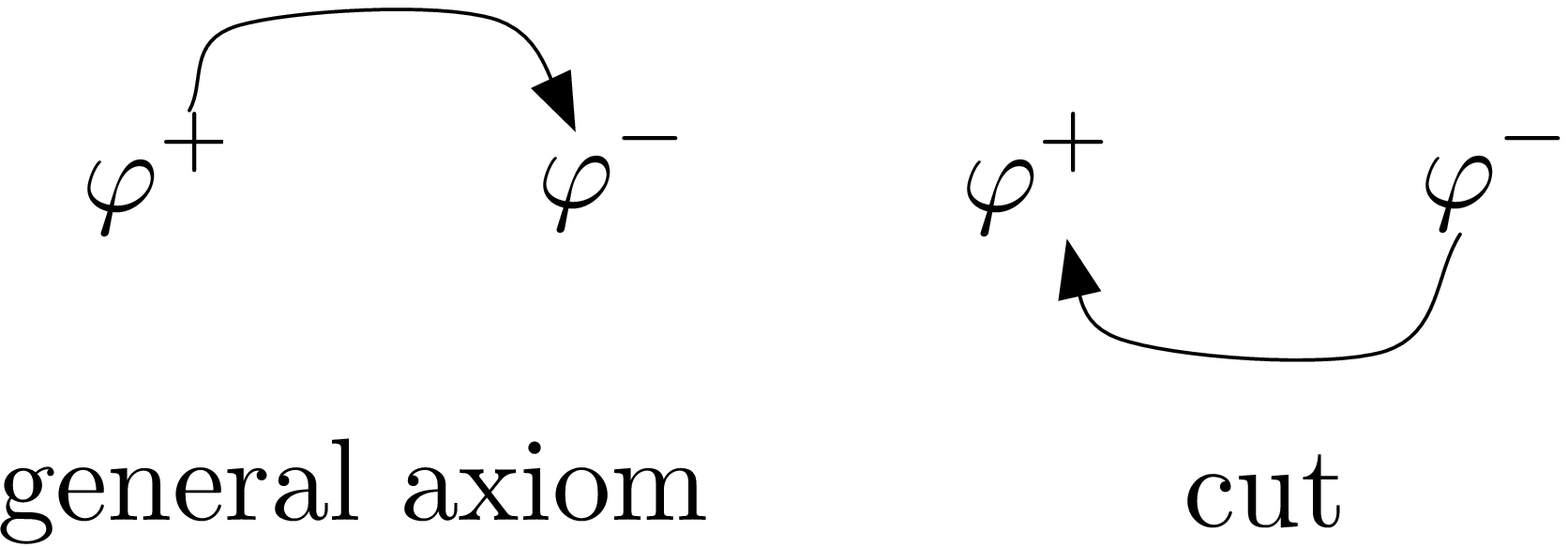}\enspace.
 \end{center}

\subsection{The Amida Nets}

 \begin{definition}[The Amida nets]
  \label{def:amidanets}
For a hypersequent~$\hyper$,
\textit{Amida nets}\index{Amida nets} of $\hyper$ are inductively
  defined by the following three clauses:
\begin{itemize}
 \item an essential net of $\hyper$ is an Amida net of $\hyper$;
 \item for an Amida net of $\hyper$ with two different\footnote{The two edges can
       be connected by a new edge as long as they are different; their
       relative positions do not matter.} up-edges,
	\begin{center}
	 \includegraphics[scale=0.3]{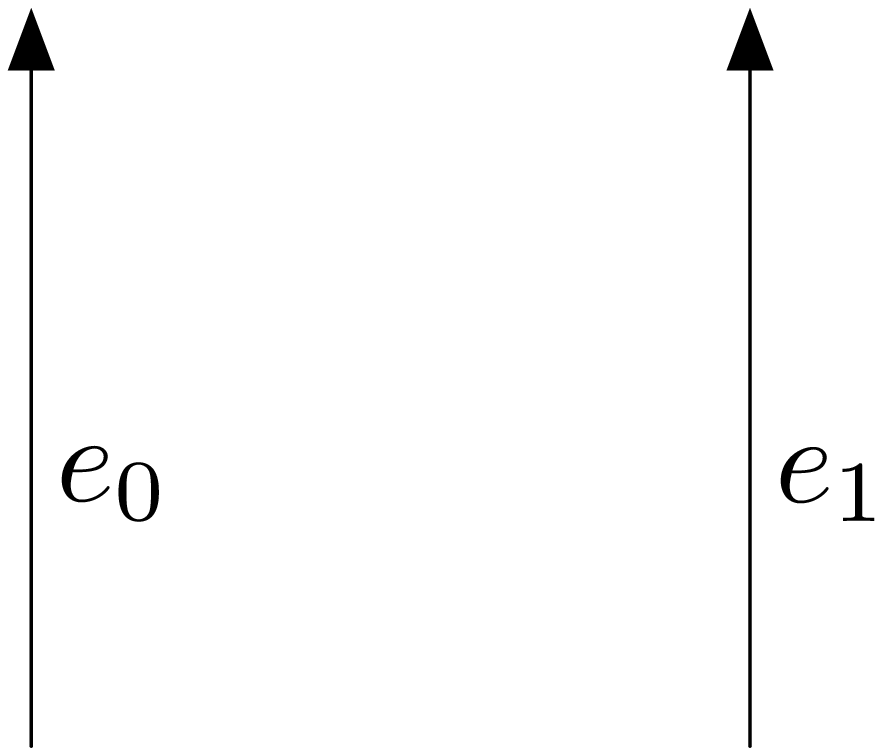}
	\end{center}
       replacing these with
	\begin{center}
	 \includegraphics[scale=0.3]{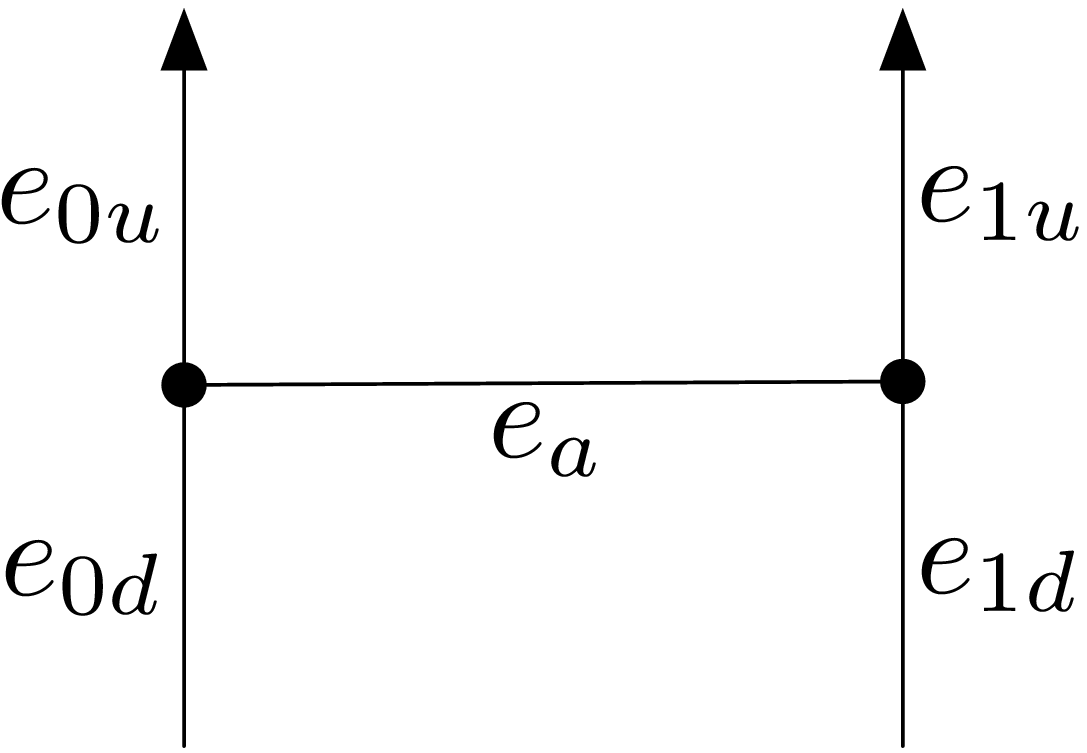}
	\end{center}
       yields an Amida net of $\hyper$,
       where the above component has two paths $e_{0d} e_a e_{1u}$
       and $e_{1d} e_a e_{0u}$;
 \item for an Amida net of $\hyper$ with an up-edge,
	\begin{center}
	 \includegraphics[scale=0.3]{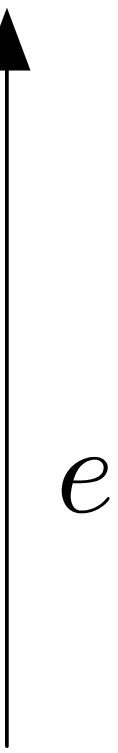}
	\end{center}
       replacing this with
	\begin{center}
	 \includegraphics[scale=0.3]{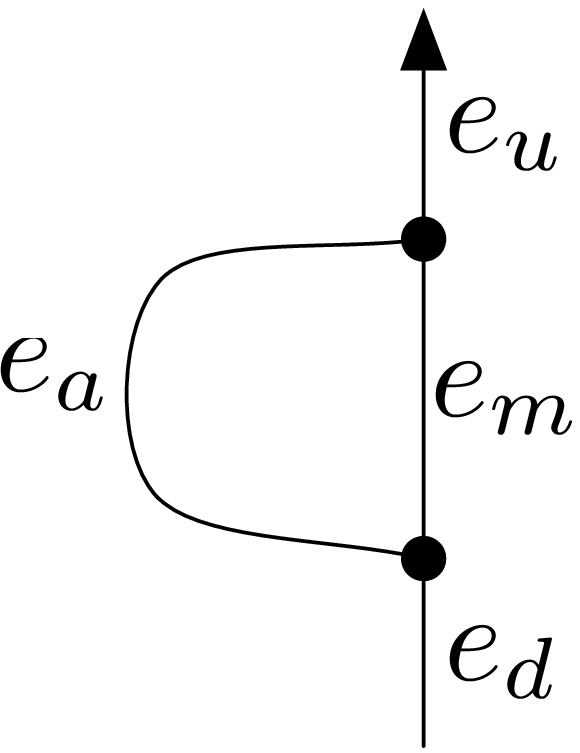}
	\end{center}
       yields an Amida net of $\hyper$,
       where the above component has
       one finite path $e_de_ae_u$
       and one infinite path $\cdots e_m e_a e_m e_a \cdots$.
\end{itemize}
 \end{definition}
 In these clauses, we call the edges labeled~$e_a$ the \textit{Amida
 edges}\index{edge!Amida}\index{Amida edge}.

 \begin{definition}[Correct Amida nets]
  A \textit{correct}\index{Amida net!correct}\index{net!correct
  Amida}\index{correct Amida net}
  Amida net is an Amida net satisfying the three conditions
  in Definition~\ref{def:amidanets}.
 \end{definition}

 The Amida edge is not merely a crossing of up-edges.
 See the difference between
 \begin{center}
\raisebox{-.5\height}{\includegraphics[scale=0.2]{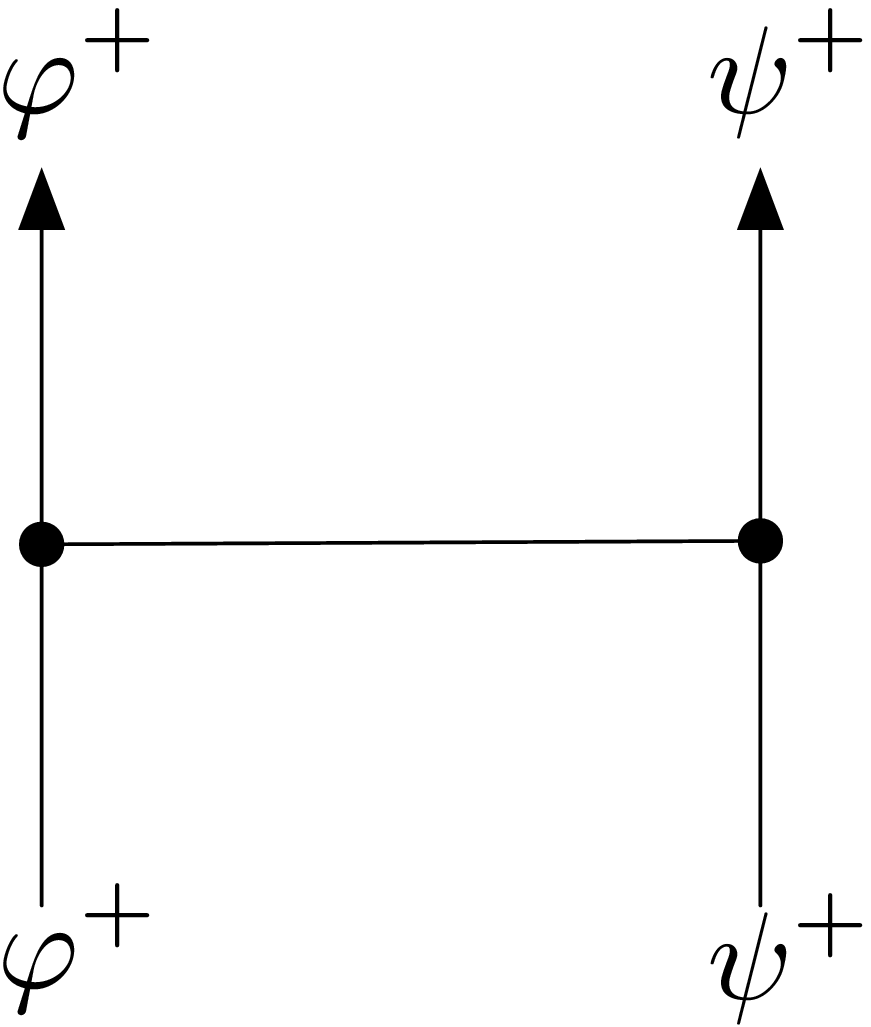}}
and
\raisebox{-.5\height}{\includegraphics[scale=0.2]{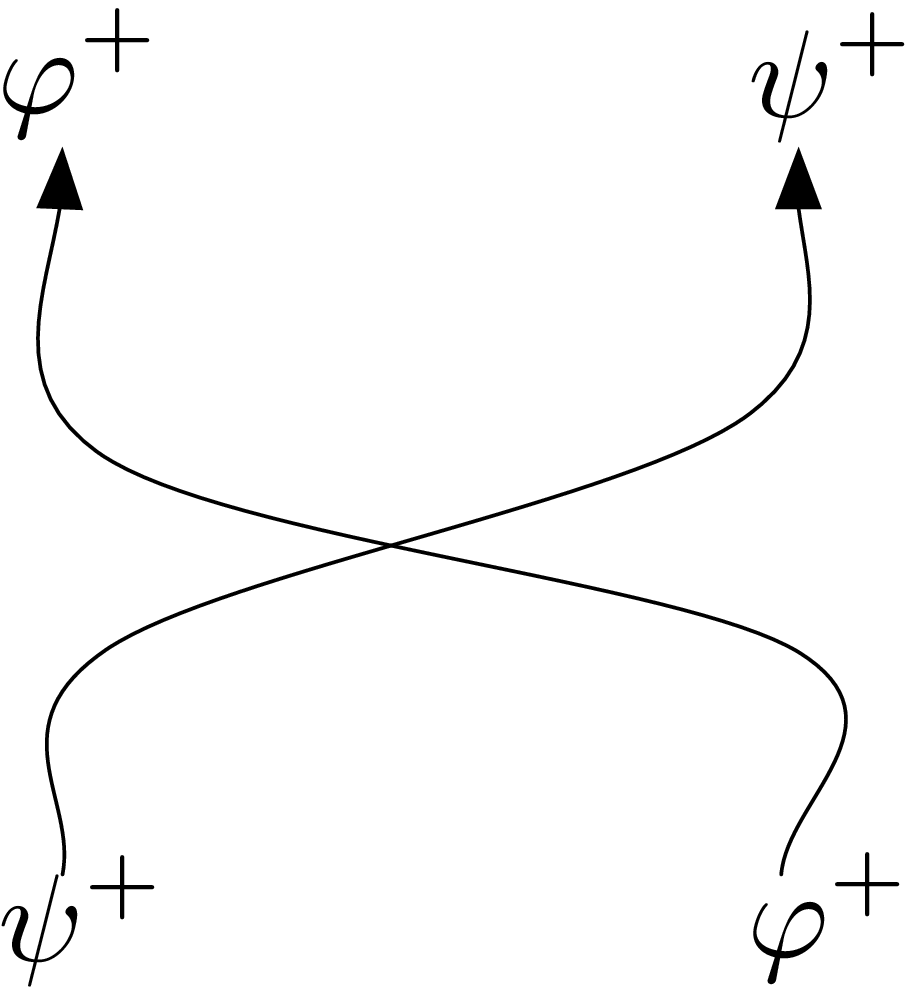}}\enspace.
 \end{center}
The difference is the labels at the bottom.
Although Amida edges cross the paths,
they do not transfer labels.
This difference of labels makes Amida nets validate the Amida axiom.
 \begin{example}[A correct Amida net for the Amida axiom]
  Here is a correct Amida net for the Amida axiom $(X\limp Y)\otimes(Y\limp X)$.
   \begin{center}
    \includegraphics[scale=0.2]{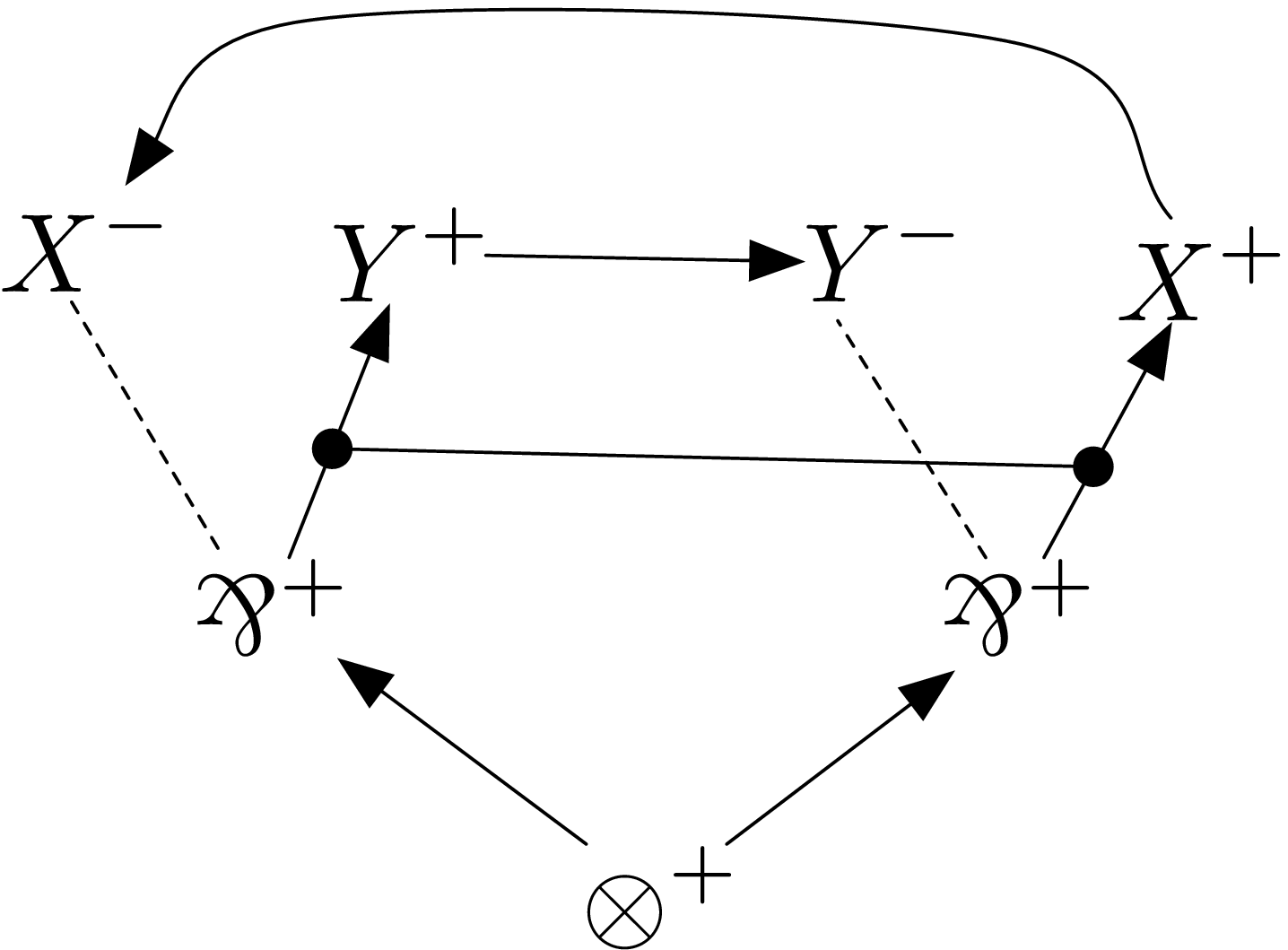}
   \end{center}
  In terms of the set of paths, the above Amida net is equivalent to the
  following correct essential net for $(X\limp X)\otimes (Y\limp Y)$.
   \begin{center}
    \includegraphics[scale=0.2]{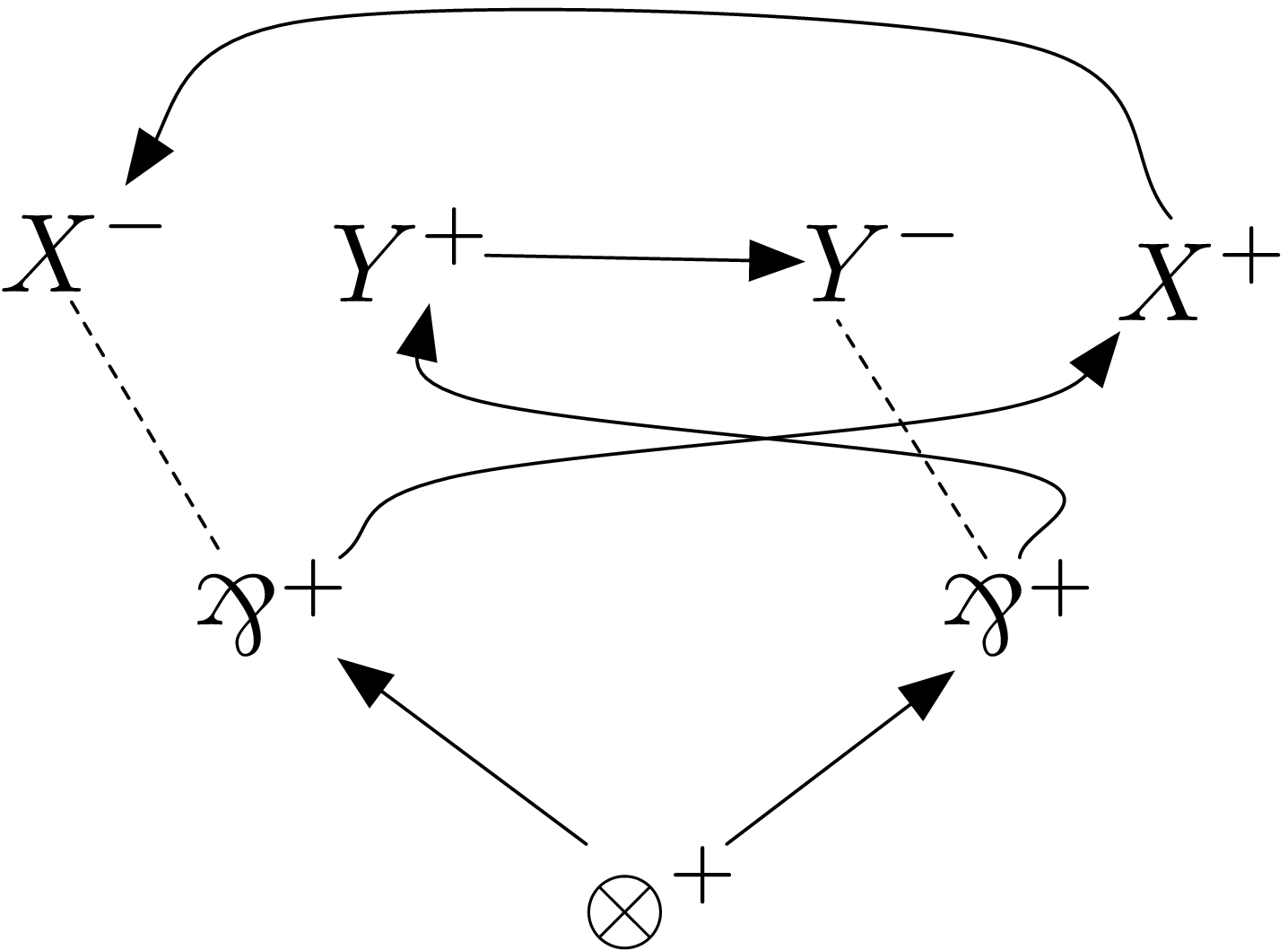}
   \end{center}
 \end{example}

\subsection{Soundness and Completeness of Amida nets}

 \begin{theorem}[Completeness of Amida nets]
  \label{amida-net-complete}
  If a hypersequent $\hyper$ is derivable,
  there is a correct Amida net for $\hyper$.
 \end{theorem}
 \begin{proof}
  Inductively on hypersequent derivations.
  The Sync rule is translated into a crossing with an Amida edge:
   \begin{center}
    \AxiomC{$\G\tr\phi\hmid\D\tr\psi$}
    \UnaryInfC{$\G\tr\psi\hmid\D\tr\phi$}
    \DisplayProof
    \quad
    $\mapsto$
    \quad
    $\vcenter{\hbox{\includegraphics[scale=0.3]{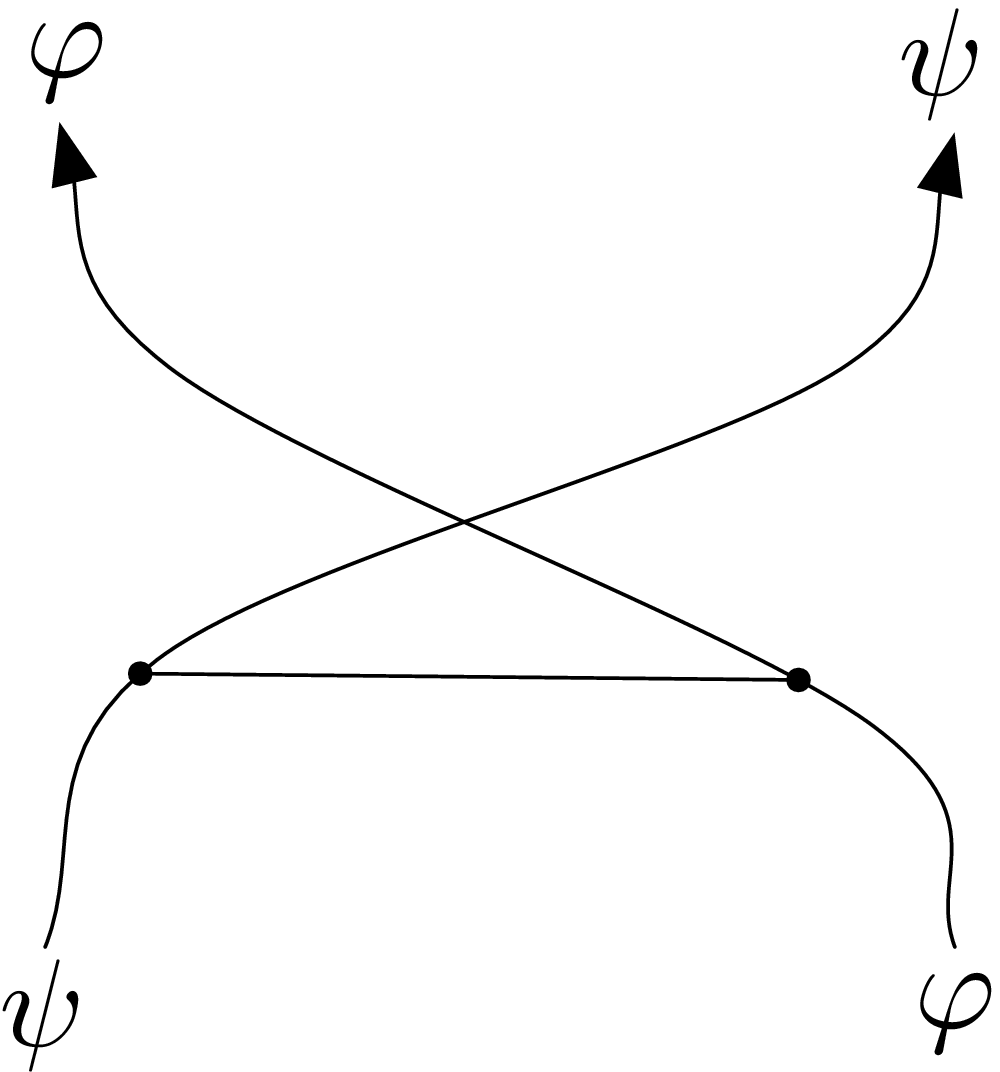}}}$
   \end{center}
  where the crossing exchanges the formulae and the Amida edge keeps the
  path connections vertically straight.
 \end{proof}
 \begin{theorem}[Soundness of Amida nets]
  If there is a correct Amida net for a hypersequent~$\hyper$, then
  $\hyper$ is derivable.
 \end{theorem}
 \begin{proof}
  From a correct Amida net, first we move the Amida edges upwards
  until they are just below axiom edges\footnote{The idea is similar to
  the most popular syntactic
  cut-elimination proofs.}.
The moves are as follows.
 \begin{center}
\includegraphics[scale=0.2]{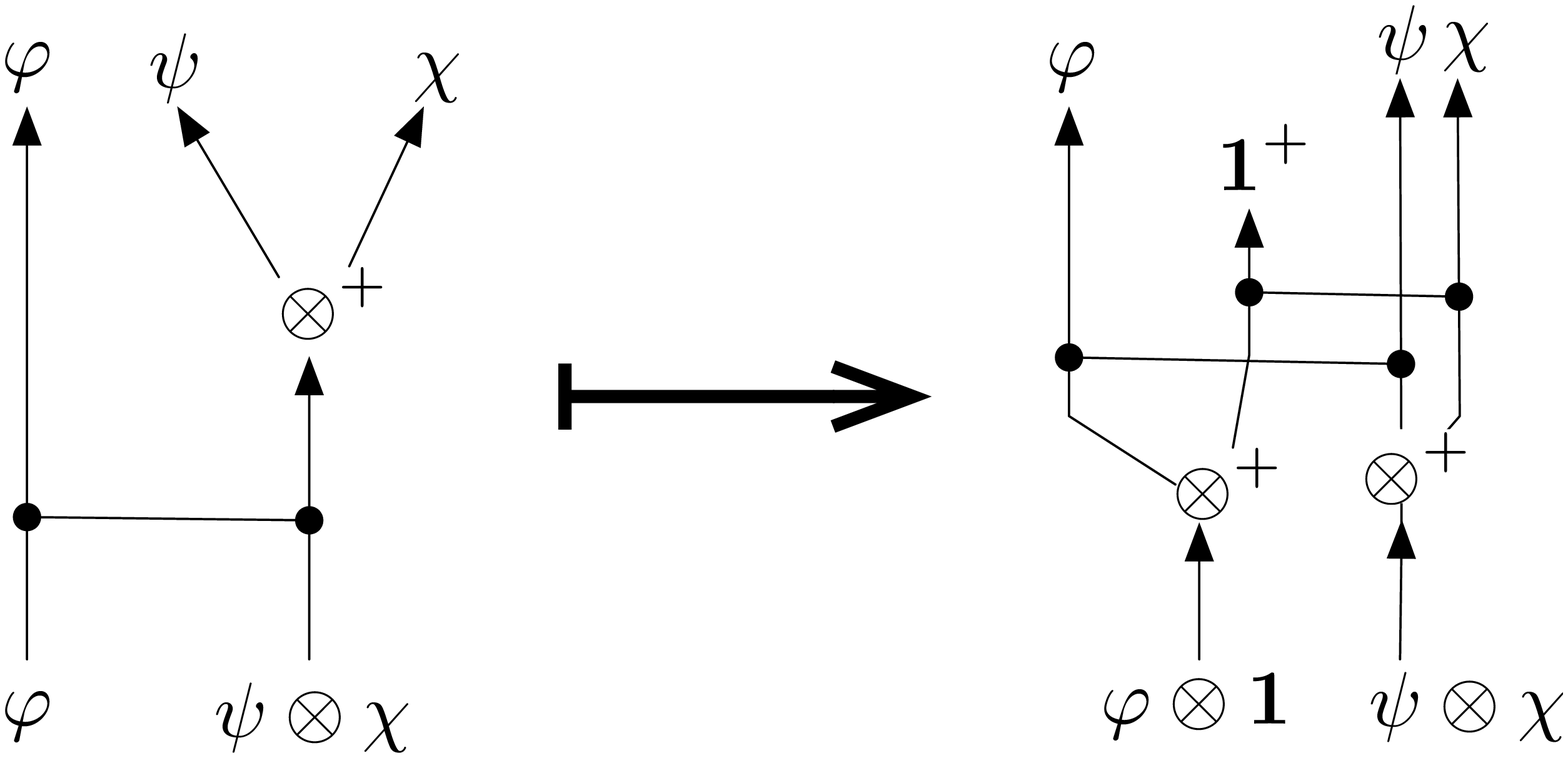}
\vskip 1.2cm
\includegraphics[scale=0.2]{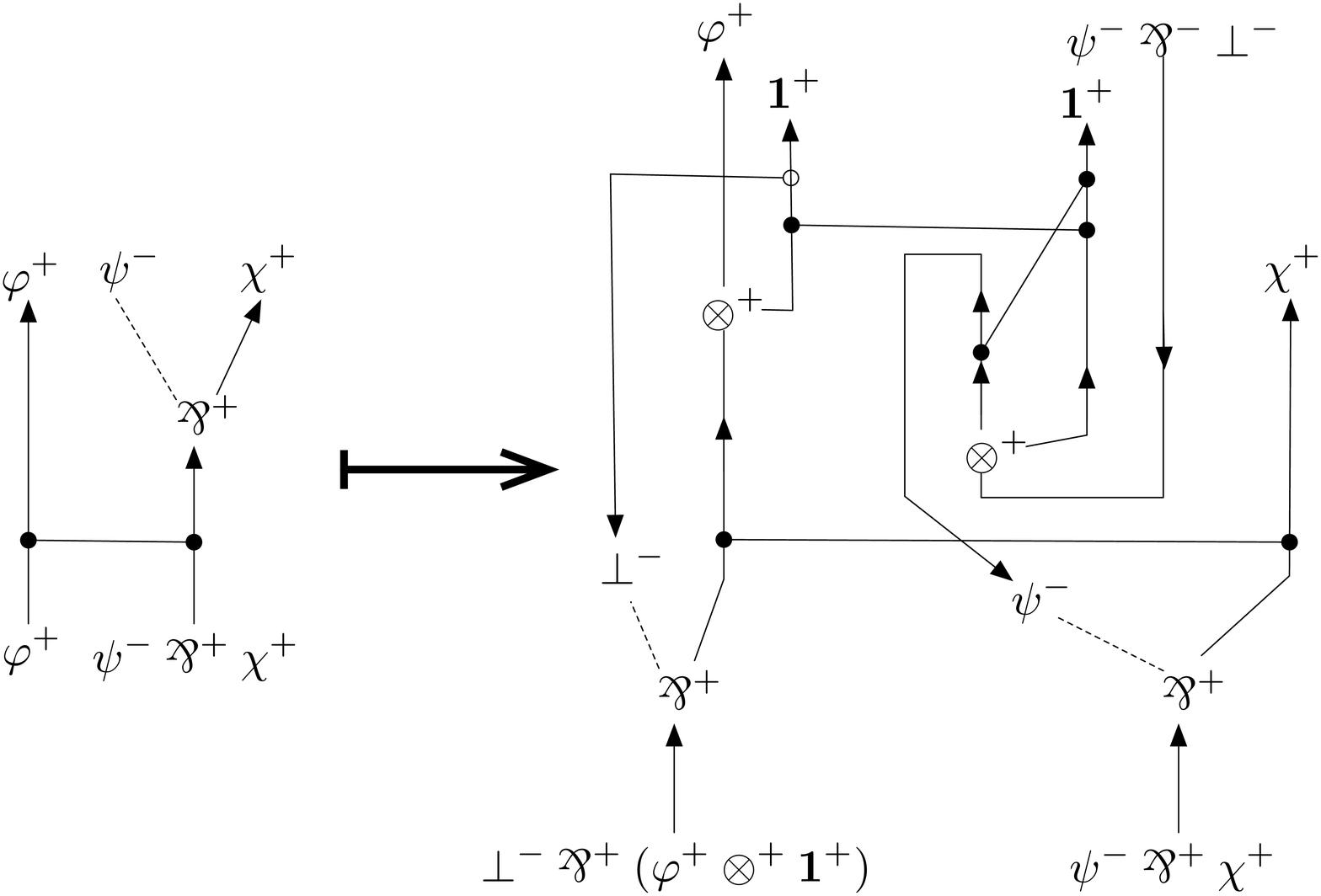}
 \end{center}
These translations have two properties.
\begin{enumerate}
 \item When the original (contained in a larger picture)
       is a correct Amida net, the translation (contained in the same
       larger picture) is also a correct Amida net.
 \item The end nodes of the translation correspond to the end nodes of
       the original, and the corresponding end nodes have the same label
       (up to logical equivalence in IMLL).  In case of $\otimes$
       translation, $\phi$ and
       $\phi\otimes \one$ are logically equivalent because $\one$ is the
       unit of $\otimes$.
       In case of $\parr$ translation, $\phi$ and $\bot\parr \phi$ are
       logically equivalent because $\bot$ is the unit of $\parr$.
\end{enumerate}
For checking the first condition, it is enough to follow the paths
(crossing all Amida edges).
For the second condition, it is enough to follow the vertical edges
ignoring the Amida edges and $\bot$-branches.

The $\otimes$ move introduces Amida edges only above the branching rules.
Although the $\parr$ move introduces an Amida edge below a
branching rule, that branching rule is of $\otimes$ nature.
Also, the $\parr$ move introduces an Amida edge below $\psi$-axiom link,
which is actually a macro.  So we have to continue applying the
translation moves in the macro.  However, since $\psi$ is a strictly
smaller subformula of $\psi\parr\chi$, this does not cause infinite
recursion.

Then, by these translation moves,
the whole Amida net is decomposed vertically into three layers.
At the top, there is a layer with only axiom edges.
In the middle, there is a layer with only vertical edges and Amida
edges.
At the bottom, there is a layer that contains only ordinary
essential net nodes.

Since the middle layer is an Amida lottery, it defines a permutation.
That permutation can be expressed as a product of transpositions, so
that
the original Amida lottery is equivalent to an encoding of a
  hypersequent derivation
that consists of only Sync rules.

After we encode the top and the middle layer into a hypersequent
derivation, encoding the bottom layer
can be done in the same way as Lamarche's approach~\cite{lamarche2008}.
\end{proof}

We wonder whether it is possible to add Amida edges to
the IMALL$^-$ essential nets following
Lamarche~\cite{lamarche2008}.
The additives are notoriously difficult for proof nets and we do not
expect the combination of additive connectives and Amida edges can be
treated in any straightforward way.

\end{document}